\documentclass[11pt]{article}
\usepackage[a4paper]{geometry}
\usepackage{amsmath,amssymb,amsthm,mathtools}
\usepackage{graphicx,booktabs,enumitem,microtype,needspace,float}
\usepackage{algorithm,algorithmic}
\usepackage{tikz}
\usetikzlibrary{arrows.meta,positioning,calc,shapes.geometric}
\usepackage[font=small,labelfont=bf]{caption}
\usepackage[hidelinks]{hyperref}
\hypersetup{pdftitle={Simple Strategy-Iteration with Dantzig's Rule is not Strongly Polynomial on Turn-Based Deterministic Games},pdfauthor={Sanyou Mei},pdfsubject={Research paper record copy}}
\allowdisplaybreaks[3]
\setlist[enumerate]{leftmargin=*,itemsep=3pt,topsep=5pt}
\newtheorem{theorem}{Theorem}
\newtheorem{corollary}{Corollary}
\newtheorem{lemma}{Lemma}
\newtheorem{proposition}{Proposition}
\newtheorem{definition}{Definition}
\newtheorem{remark}{Remark}

\newcommand{\rr}{\mathbb R}
\newcommand{\mcS}{\mathcal S}
\newcommand{\mcA}{\mathcal A}
\newcommand{\mcG}{\mathcal G}
\newcommand{\bfc}{\mathbf c}
\newcommand{\bfP}{\mathbf P}
\newcommand{\bfp}{\mathbf p}
\newcommand{\bfv}{\mathbf v}
\newcommand{\bfe}{\mathbf e}
\newcommand{\bfb}{\mathbf b}

\newcommand{\bfu}{\mathbf u}
\newcommand{\bfzero}{\mathbf 0}
\DeclareMathOperator*{\Argmin}{Argmin}
\title{Simple Strategy-Iteration with Dantzig's Rule is not Strongly Polynomial on Turn-Based Deterministic Games}
\author{Sanyou Mei\thanks{Department of Industrial Engineering and Decision Analytics, Hong Kong University of Science and Technology, Hong Kong, China (email: \texttt{symei@ust.hk}).}}
\date{October 19, 2026}

\begin{document}
\maketitle

\begin{abstract}
We study turn-based deterministic zero-sum games with discounted and long-run average costs. For discounted costs, we establish exponential and superpolynomial iteration lower bounds for the simple strategy-iteration method with Dantzig's rule under unrestricted costs and unit weights, respectively. To the best of our knowledge, we establish for the first time that the method is not strongly polynomial, and this remains true even when all action costs belong to $\{-1,0,1\}$. For long-run average costs, we extend the Dantzig counterpart in~\cite{zhongye2026} to the two-player setting and prove that it terminates at an optimal strategy profile. However, a similar counterexample construction establishes an exponential iteration lower bound for termination with unrestricted costs. This construction does not extend to unit weights, for which we prove that the algorithm admits a strongly polynomial implementation.
\end{abstract}



\section{Introduction}\label{sec:intro}

A turn-based deterministic zero-sum game (TBDG) is a special case of a turn-based stochastic zero-sum game (TBSG), in which every action has a deterministic successor. These games belong to the stochastic-game framework introduced by Shapley~\cite{shapley1953}. The states are partitioned between two players. At each state, the controlling player chooses an action, incurs its immediate cost, and moves to the corresponding successor. Player~1 minimizes the total discounted cost, while Player~2 maximizes it. When all states are controlled by one player, the game reduces to a deterministic Markov decision process (MDP).

Discounted MDPs can be solved in polynomial time through linear programming. Strongly polynomial bounds have also been established under additional assumptions. For a fixed discount factor, Ye~\cite{ye2005,ye2011} proved that an interior-point method, the simplex method with Dantzig's rule, and Howard's policy-iteration method are strongly polynomial. Scherrer~\cite{scherrer2016} sharpened the policy-iteration bounds and analyzed the effects of restrictions on transient and recurrent behavior. For deterministic MDPs, Post and Ye~\cite{postye2015} established a strongly polynomial bound for the simplex method with Dantzig's rule that is independent of the discount factor, and Hansen, Kaplan, and Zwick~\cite{hansenkaplanzwick2014} obtained a sharper bound. Thus, deterministic MDPs can be solved in strongly polynomial time even when the discount factor is part of the input.

For general TBDGs, however, the existence of a strongly polynomial algorithm remains open. Considerable work has been devoted to developing efficient algorithms for TBDGs, including value iteration, interior-point methods based on linear complementarity formulations, and strategy iteration. Kozachinskiy~\cite{kozachinskiy2021} developed a polyhedral value-iteration algorithm with a $2^{O(n)}$ arithmetic-operation bound independent of the discount factor. Hansen and Ibsen-Jensen~\cite{hansenibsenjensen2013} studied interior-point methods through a reduction to $P$-matrix linear complementarity problems, obtaining complexity bounds that depend on the discount factor. For a fixed discount factor, Hansen, Miltersen, and Zwick~\cite{hansenmiltersenzwick2013} proved strong polynomiality of strategy iteration for TBSGs. Jia, Wen, and Ye~\cite{jiawenye2020} analyzed simple strategy iteration with Dantzig's rule. At each iteration, this method changes one action of Player~1 with the most negative reduced cost and then computes an optimal counter strategy of Player~2. Its iteration bound depends on the discount factor and therefore does not establish strong polynomiality when that factor is part of the input.

Further results have been obtained for some special cases. For discounted TBDGs with unit weights, namely, costs in $\{-1,0,1\}$, Asadi et al.~\cite{asadietal2024} established a deterministic $n^{O(\sqrt n)}$ bound for strategy iteration, independently of the discount factor. For mean-payoff games, whose objective is the long-run average cost, Zwick and Paterson~\cite{zwickpaterson1996} and Brim et al.~\cite{brimetal2011} developed pseudopolynomial algorithms. For mean-payoff games with unit weights, these pseudopolynomial bounds yield strongly polynomial algorithms. Bj\"orklund and Vorobyov~\cite{bjorklundvorobyov2007} obtained a randomized strongly subexponential strategy-improvement algorithm for mean-payoff games. More recently, Mei, Sun, and Ye~\cite{meisunye2026} established a strongly polynomial bound for simple strategy iteration with Dantzig's rule on turn-based deterministic forward games. This subclass excludes directed cycles involving both players, while allowing cycles controlled entirely by either player. Together with the results for deterministic MDPs and fixed-discount games discussed above, this result makes simple strategy iteration with Dantzig's rule a particularly promising approach to strongly polynomial solvability of TBDGs. It is therefore natural to ask whether the method remains strongly polynomial when mixed-player cycles are allowed, particularly under the unit-weight restriction.

Alongside these algorithmic developments, considerable effort has been devoted to establishing lower bounds on the complexity of these methods. For MDPs, Hollanders, Delvenne, and Jungers~\cite{hollandersetal2012} proved an exponential lower bound for Howard's policy iteration with stochastic transitions and a discount factor supplied as part of the input. Zhong and Ye~\cite{zhongye2026} established exponential lower bounds for Howard's policy iteration on deterministic discounted MDPs. Disser and Mosis~\cite{dissermosis2023} obtained exponential lower bounds for Dantzig's, Bland's, and Largest Increase rules on MDPs, but their construction for Dantzig's rule uses stochastic transitions. Disser et al.~\cite{disseretal2026} also studied lower bounds for ranking-based pivot rules. For two-player games, Friedmann~\cite{friedmann2011} proved exponential lower bounds for locally and globally optimizing strategy-improvement rules on simple stochastic games and deterministic discounted games. Avis and Friedmann~\cite{avisfriedmann2017} established an exponential lower bound for Cunningham's cyclic single-switch rule on deterministic discounted games. However, these lower bounds concern different update rules or transition models and do not determine whether simple strategy iteration with Dantzig's rule and exact optimal counter strategies is strongly polynomial on TBDGs. This leaves a significant gap between the positive results for this method on restricted classes and our understanding of its complexity on general TBDGs.

In this paper, we construct counterexamples showing that simple strategy iteration with Dantzig's rule is not strongly polynomial on general discounted TBDGs, even with unit weights. We also establish an exponential iteration lower bound for termination of a Dantzig counterpart on mean-payoff TBDGs with unrestricted costs. The main idea is to make Player~1 create a mixed-player cycle and Player~2 cancel it through the optimal counter strategy at every iteration. By choosing the costs that determine state-value ordering separately from those that determine pivot selection, we force the method to follow a binary counter. This prevents the cycle-creation arguments used for deterministic MDPs (Post and Ye~\cite{postye2015}) and deterministic forward games (Mei, Sun, and Ye~\cite{meisunye2026}) from extending directly to general TBDGs. Our contributions are as follows.
\begin{itemize}
\item With unrestricted costs, we construct a TBDG with $n=4h+2$ states and $m=8h+2$ actions on which simple strategy-iteration with Dantzig's rule performs exactly $2^h-1$ iterations from the specified initial strategy. This establishes an exponential iteration lower bound of $2^{\Omega(n)}$.
\item With unit weights, we construct a TBDG with $n=\Theta(h^5)$ states and $m=n+4h$ actions that preserves the same $2^h-1$ iterations. The discount factor has $O(h^3\log h)$ binary encoding length. Hence the iteration lower bound is superpolynomial, namely, $2^{\Omega(n^{1/5})}$.
\item For TBDGs with average costs, we extend the Dantzig counterpart of Zhong and Ye~\cite{zhongye2026} and prove that it terminates at an optimal strategy profile. With unrestricted costs, we extend the counterexample for discounted costs to establish an exponential iteration lower bound for termination. With unit weights, it has a strongly polynomial complexity.
\end{itemize}
To the best of our knowledge, these are the first exponential and superpolynomial iteration lower bounds for this simple strategy-iteration method with Dantzig's rule on discounted TBDGs with unrestricted costs and unit weights, respectively.

The rest of the paper is organized as follows. Section~\ref{sec:prelim} introduces the notation and the simple strategy-iteration method. Section~\ref{sec:real} presents the construction with unrestricted costs and proves the exponential lower bound. Section~\ref{sec:ternary} extends the construction to unit weights. Section~\ref{sec:mean-payoff} studies a Dantzig counterpart for long-run average costs, establishes its convergence, and proves an exponential iteration lower bound with unrestricted costs and strongly polynomial bounds with unit weights. Section~\ref{sec:conclusion} concludes the paper.

\section{Preliminaries and the simple strategy-iteration method}\label{sec:prelim}

We use the following notation throughout the paper. Let $\rr^n$ denote the $n$-dimensional Euclidean space and $\mathbb Z_{\geq0}$ the set of nonnegative integers. For every positive integer $n$, let $[n]=\{1,\ldots,n\}$. For a finite set $S$, let $|S|$ denote its cardinality. Let $\bfe$, $\bfzero$, $\bfe_s$, and $I$ denote the all-ones vector, the all-zeros vector, the $s$th coordinate vector, and the identity matrix, respectively, with dimensions clear from the context. The vector inequalities are understood componentwise. For $t\in\rr$, let $[t]_-:=\min\{t,0\}$, and let $\operatorname{sgn}t$ be $-1$, $0$, or $1$ according as $t<0$, $t=0$, or $t>0$. For integers $j,k\geq0$, let $\binom jk$ denote the number of $k$-element subsets of a $j$-element set. For a polynomial $F$, let $F^{(k)}$ denote its $k$th derivative. To proceed, we follow the TBDG notation of \cite{meisunye2026}.
\begin{definition}[TBDG]\label{def:tbdg}
A turn-based deterministic zero-sum game is specified by $\mcG=(\mcS,\mcA,\bfc,\bfP,\gamma)$, where the following hold.
\begin{enumerate}[label=(\roman*)]
\item $\mcS=\mcS^1\cup\mcS^2=[n]$ is the state space, partitioned between Players~1 and~2.
\item $\mcA=\bigcup_{s\in\mcS}\mcA_s=[m]$ is the action space. The nonempty sets $\mcA_s$ are pairwise disjoint, and $\mcA_s$ contains the actions available at state $s$. Let $\mcA^i=\bigcup_{s\in\mcS^i}\mcA_s$ for $i=1,2$.
\item $\bfc\in\rr^m$ is the cost vector, with immediate cost $c_a$ for action $a\in\mcA$.
\item $\bfP\in\rr^{n\times m}$ has column $\bfp_a$ for each action $a$. Each $\bfp_a$ is a point-mass probability vector. Thus, if $a\in\mcA_s$ moves from $s$ to $t$, then $\bfp_a=\bfe_t$.
\item $\gamma\in(0,1)$ is the common discount factor. Player~1 minimizes the infinite-horizon discounted cost and Player~2 maximizes it.
\end{enumerate}
\end{definition}

The underlying directed graph of $\mcG$ is
\[
G=(\mcS,E),\qquad
E=\{(s,t)\in\mcS\times\mcS:\bfp_a=\bfe_t\text{ for some }a\in\mcA_s\}.
\]
For an action $a\in\mcA_s$ with $\bfp_a=\bfe_t$, the ordered pair $(s,t)$ is its associated directed edge. If the endpoints determine the action uniquely, as in our constructions, we identify $a$ with $(s,t)$ and write $c_{(s,t)}$ and $r_{(s,t)}^\pi$ for its cost and reduced cost, respectively. Actions with the same endpoints remain distinct elements of $\mcA$, so $m=|\mcA|$ need not equal $|E|$. Self-loops are allowed. A mixed-player cycle is a directed cycle containing states controlled by both players.

A strategy $\sigma$ of Player~1 chooses one action $\sigma_s\in\mcA_s$ at every $s\in\mcS^1$; a strategy $\tau$ of Player~2 is defined analogously on $\mcS^2$. A strategy profile is $\pi=(\sigma,\tau)$, with selected action $\pi_s$ at state $s$. Its induced graph is $G^\pi=(\mcS,\{(s,t)\in\mcS\times\mcS:\bfp_{\pi_s}=\bfe_t\})$, so every state has exactly one outgoing edge. State and action labels in a construction may be identified with $[n]$ and $[m]$ by any fixed enumeration.

For a profile $\pi$, let $P^\pi$ be the $n\times n$ matrix whose column $s$ is $\bfp_{\pi_s}$, and let $c_s^\pi=c_{\pi_s}$. Thus $P^\pi_{t,s}$ is the probability of moving from $s$ to $t$. The cost-to-go value vector $\bfv^\pi$ is the unique solution of
\begin{equation}\label{eq:value}
\bfv^\pi=\bfc^\pi+\gamma(P^\pi)^T\bfv^\pi,
\qquad
\bfv^\pi=(I-\gamma P^\pi)^{-T}\bfc^\pi.
\end{equation}
In particular, $v_s^\pi=c_{\pi_s}+\gamma\bfp_{\pi_s}^T\bfv^\pi$.

For an action $a\in\mcA_s$, its reduced cost under $\pi$ is
\begin{equation}\label{eq:reduced}
r_a^\pi=c_a+\gamma\bfp_a^T\bfv^\pi-v_s^\pi.
\end{equation}
For $a\in\mcA_s$ with $\bfp_a=\bfe_t$, this becomes $r_a^\pi=c_a+\gamma v_t^\pi-v_s^\pi$. Every selected action has reduced cost zero. A negative reduced cost is improving for Player~1, and a positive reduced cost is improving for Player~2. We also use the normalization $\bfb=(1-\gamma)\bfe$. Then $\bfb^T\bfv^\pi$ is the scaled total cost-to-go value. The individual values and reduced costs remain unscaled.

\begin{definition}[Optimal counter strategy]\label{def:response}
For a fixed strategy $\sigma$ of Player~1, an optimal counter strategy $\pi^2(\sigma)$ of Player~2 satisfies
\[
\bfv^{(\sigma,\tau)}\leq\bfv^{(\sigma,\pi^2(\sigma))}
\quad\text{for every strategy $\tau$ of Player~2}.
\]
Player~1's optimal counter strategy $\pi^1(\tau)$ is defined analogously by minimization. A profile is optimal if each player's strategy is an optimal counter strategy against the other player's strategy.
\end{definition}

An optimal strategy profile exists for every TBDG. As in \cite[Subsection~1.2]{meisunye2026}, its cost-to-go value vector $\bfv^*$ and corresponding optimal strategy profile $\pi^*$ satisfy
\[
\begin{aligned}
v_s^*&=\min_{a\in\mcA_s}\{c_a+\gamma\bfp_a^T\bfv^*\},\quad\pi_s^*\in\mathop{\rm argmin}_{a\in\mcA_s}\{c_a+\gamma\bfp_a^T\bfv^*\},\quad\forall s\in\mcS^1,\\
v_s^*&=\max_{a\in\mcA_s}\{c_a+\gamma\bfp_a^T\bfv^*\},\quad\pi_s^*\in\mathop{\rm argmax}_{a\in\mcA_s}\{c_a+\gamma\bfp_a^T\bfv^*\},\quad\forall s\in\mcS^2.
\end{aligned}
\]

Let $\nu^*=\bfb^T\bfv^*$. The following proposition recalls Proposition~1(vii) and Lemma~3 of \cite{meisunye2026}, including the corresponding optimality criterion. The proof is standard in the literature.

\begin{proposition}\label{prop:reduced-optimality}
Denote by $\nu^1$ the objective value $\bfb^T\bfv^{(\sigma,\pi^2(\sigma))}$ of any strategy profile $(\sigma,\pi^2(\sigma))$, and denote by $\nu^2$ the objective value $\bfb^T\bfv^{(\pi^1(\tau),\tau)}$ of any strategy profile $(\pi^1(\tau),\tau)$. Then $\nu^2\leq\nu^*\leq\nu^1$. Moreover, the corresponding reduced-cost vector satisfies
\[
r^{(\sigma,\pi^2(\sigma))}_a\leq0\quad\forall a\in\mcA^2\quad\mbox{and}\quad r^{(\pi^1(\tau),\tau)}_a\geq0\quad\forall a\in\mcA^1,
\]
and $r^\pi_{\pi_s}=0$ for every state $s\in\mcS$ and every strategy profile $\pi$. Furthermore, a strategy profile $\pi$ is optimal if and only if
\[
r_a^\pi\geq0\quad\forall a\in\mcA^1\quad\mbox{and}\quad r_a^\pi\leq0\quad\forall a\in\mcA^2.
\]
\end{proposition}

We consider the following simple strategy-iteration method. It uses the update rule of Algorithm~1 in~\cite{jiawenye2020}, with the signs reversed because Player~1 minimizes, and is initialized with an optimal counter strategy. Ties are broken arbitrarily.

\begin{algorithm}[H]
\caption{Simple strategy iteration with Dantzig's rule}\label{alg:ssi}
\begin{algorithmic}[1]
\REQUIRE A strategy $\sigma_0$ of Player~1.
\STATE Compute $\tau_0=\pi^2(\sigma_0)$ and $\bfv^{\pi_0}$, where $\pi_0=(\sigma_0,\tau_0)$. Set $k=0$.
\STATE Choose $a'_k\in\Argmin_{a\in\mcA^1}r_a^{\pi_k}$. If $r_{a'_k}^{\pi_k}=0$, terminate and return $\pi_k$.
\STATE Replace the selected action at the source of $a'_k$ by $a'_k$, obtaining $\sigma_{k+1}$.
\STATE Compute $\tau_{k+1}=\pi^2(\sigma_{k+1})$ and $\bfv^{\pi_{k+1}}$, where $\pi_{k+1}=(\sigma_{k+1},\tau_{k+1})$. Set $k\leftarrow k+1$ and return to Step~2.
\end{algorithmic}
\end{algorithm}

\section{An exponential lower bound with unrestricted costs}\label{sec:real}

The strongly polynomial analyses for deterministic MDPs in~\cite{postye2015} and turn-based deterministic forward games in~\cite{meisunye2026} rely on the progress obtained through cycle creation. For general TBDGs, however, a mixed-player cycle created by a Player~1 pivot may be canceled by Player~2's optimal counter strategy. Motivated by this observation, we construct a family in which such cancellation occurs at every iteration. Costs on actions to the next level determine the ordering of state values, while costs on return actions determine Dantzig's pivot selection. These choices force the algorithm to follow a binary counting sequence and require exponentially many iterations. We first present the construction with unrestricted costs, for which the required cost scales can be assigned directly to actions.

\subsection{The construction of the $h$-level counterexample}\label{subsec:construction}

Fix an integer $h\geq2$, we construct an $h$-level counterexample $\mcG_h$ with states $a_i,b_i$ for $1\leq i\leq h+1$ and auxiliary states $p_i,q_i$ for $1\leq i\leq h$. For each $1\leq i\leq h$, state $a_i$ has two zero-cost actions leading to $p_i$ and $q_i$. States $p_i$ and $q_i$ have actions leading to $a_{i+1}$ and $b_{i+1}$, respectively, each with cost $W_i/\gamma$. Both auxiliary states also have an action leading back to $a_i$ with cost $-K_i/\gamma$. State $b_i$ has two zero-cost actions leading to $a_{i+1}$ and $b_{i+1}$. At the terminal level, $a_{h+1}$ has a single zero-cost self-loop, and $b_{h+1}$ has a single action leading to $a_{h+1}$ with cost $-R$. Player~1 controls $a_1,\ldots,a_{h+1}$, and Player~2 controls all other states. Choose the parameters as follows:
\begin{equation}\label{eq:real-parameters}
R=1,\qquad W_i=\gamma4^{\,i-h-1},\qquad K_i=\gamma2^{\,2-i-2h},\qquad
\delta=\frac{1}{100\times8^{h-2}},\qquad\gamma=1-\delta.
\end{equation}
No other states or actions are present. Table~\ref{tab:edges} summarizes the construction.

\begin{table}[H]
\centering
\begin{tabular}{lcccl}
\toprule
State & Owner & Actions & Successor & Cost\\
\midrule
$a_i$ & Player~1 & $2$ & $p_i;q_i$ & $0$; $0$\\
$p_i$ & Player~2 & $2$ & $a_{i+1}$; $a_i$ & $W_i/\gamma$; $-K_i/\gamma$\\
$q_i$ & Player~2 & $2$ & $b_{i+1}$; $a_i$ & $W_i/\gamma$; $-K_i/\gamma$\\
$b_i$ & Player~2 & $2$ & $a_{i+1},b_{i+1}$ & $0$; $0$\\
$a_{h+1}$ & Player~1 & $1$ & $a_{h+1}$ & $0$\\
$b_{h+1}$ & Player~2 & $1$ & $a_{h+1}$ & $-R$\\
\bottomrule
\end{tabular}
\caption{$h$-level counterexample with unrestricted costs.}
\label{tab:edges}
\end{table}

To summarize, each of the $4h$ states $a_i,b_i,p_i,q_i$, $1\leq i\leq h$, has two actions, while each terminal state $a_{h+1},b_{h+1}$ has one. Thus the game has $n=4h+2$ states and $m=8h+2$ actions. At each auxiliary state $p_i$ or $q_i$, Player~2 can either move to its designated successor at the next level or return to $a_i$, as illustrated in Figure~\ref{fig:branch}.
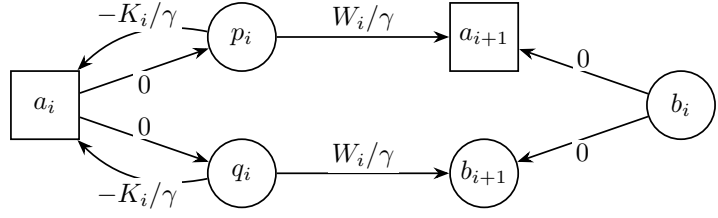
\begin{figure}[H]
\centering
\begin{tikzpicture}[>=Stealth,line width=.65pt,
every node/.style={font=\small},
state/.style={draw,minimum width=9mm,minimum height=9mm,inner sep=0pt},
p1/.style={state,rectangle},
p2/.style={state,circle},
aux/.style={p2},
pathL/.style={->},
pathB/.style={->},
cost/.style={fill=white,inner sep=1.5pt}]
\node[p1] (a) at (0,-.9) {$a_i$};
\node[aux] (p) at (2.6,0) {$p_i$};
\node[aux] (q) at (2.6,-1.8) {$q_i$};
\node[p1] (anext) at (5.8,0) {$a_{i+1}$};
\node[p2] (bnext) at (5.8,-1.8) {$b_{i+1}$};
\node[p2] (b) at (8.4,-.9) {$b_i$};
\draw[->] (a) -- node[cost,below] {$0$} (p);
\draw[->] (a) -- node[cost,above] {$0$} (q);
\draw[pathL] (p) -- node[cost,above] {$W_i/\gamma$} (anext);
\draw[pathL] (q) -- node[cost,above] {$W_i/\gamma$} (bnext);
\draw[pathL] (p) to[bend right=28] node[cost,above] {$-K_i/\gamma$} (a);
\draw[pathL] (q) to[bend left=28] node[cost,below] {$-K_i/\gamma$} (a);
\draw[pathB] (b) -- node[cost,above] {$0$} (anext);
\draw[pathB] (b) -- node[cost,below] {$0$} (bnext);
\end{tikzpicture}
\caption{The two successor branches at level $i$. Throughout this paper, squares denote Player~1 states and circles denote Player~2 states.}
\label{fig:branch}
\end{figure}
During our call to Algorithm \ref{alg:ssi}, we will initialize Player~1's initial strategy by setting
\begin{equation}\label{eq:initial}
(\sigma_0)_{a_i}=(a_i,p_i)\qquad \forall1\leq i\leq h,\quad (\sigma_0)_{a_{h+1}}=(a_{h+1},a_{h+1}).
\end{equation}

To simplify the analysis and facilitate the extension to unit weights, let $\beta=\gamma^2$ denote the total discount along either path $a_i\to p_i\to a_{i+1}$ or $a_i\to q_i\to b_{i+1}$, and let $\eta=\gamma$ denote the total discount from $b_i$ to either $a_{i+1}$ or $b_{i+1}$. Both return cycles $a_i\to p_i\to a_i$ and $a_i\to q_i\to a_i$ also have total discount $\beta$. With this notation, the parameters in \eqref{eq:real-parameters} satisfy
\begin{align}
&\frac89<\beta\leq\eta<1,\qquad W_i,K_i,R>0,\qquad K_i\leq2W_i\quad\forall1\leq i\leq h,
\label{eq:separation-a}\\
&W_{i+1}\geq4W_i,\qquad K_i\geq2K_{i+1}\quad\forall1\leq i<h,
\label{eq:separation-b}\\
&R\geq4W_h,\qquad3(1-\beta)R<K_h,\qquad4(\eta-\beta)R<W_1.
\label{eq:separation-c}
\end{align}

Formally, the construction defines $\mcG_h=(\mcS,\mcA,\bfc,\bfP,\gamma)$ with
\[
\begin{aligned}
\mcS&=\{a_i,b_i:1\leq i\leq h+1\}\cup\{p_i,q_i:1\leq i\leq h\},\\
\mcS^1&=\{a_i:1\leq i\leq h+1\},\qquad \mcS^2=\mcS\setminus\mcS^1.
\end{aligned}
\]
For each $s\in\mcS$, let $\mcA_s$ consist of the outgoing actions listed for $s$ in Table~\ref{tab:edges}, and set $\mcA=\bigcup_{s\in\mcS}\mcA_s$. Each action $a=(s,t)$ has the cost $c_a$ specified in that table and transition column $\bfp_a=\bfe_t$ in $\bfP$. The costs and the discount factor are determined by \eqref{eq:real-parameters}. As in Section~\ref{sec:prelim}, the states and actions may be enumerated by $[n]$ and $[m]$ without changing the game.

\subsection{Counter strategies, values, and reduced costs properties}\label{subsec:comparisons}

For a strategy profile $\pi=(\sigma,\tau)$, write $\Delta_i^\pi=v_{a_i}^\pi-v_{b_i}^\pi$ for $1\leq i\leq h+1$.

By construction of $\mcG_h$ in Subsection~\ref{subsec:construction}, one has $\pi_{a_{h+1}}=(a_{h+1},a_{h+1})$ and $\pi_{b_{h+1}}=(b_{h+1},a_{h+1})$ for every strategy profile $\pi$, as these are the only actions available. It then follows from \eqref{eq:value} and the definition of $\Delta_{h+1}^\pi$ that
\begin{equation}\label{eq:terminal-values}
v_{a_{h+1}}^\pi=0,\qquad v_{b_{h+1}}^\pi=-R,\qquad \Delta_{h+1}^\pi=v_{a_{h+1}}^\pi-v_{b_{h+1}}^\pi=R.
\end{equation}

\begin{lemma}\label{lem:existence}
Consider the TBDG $\mcG_h$ defined in Subsection~\ref{subsec:construction}. For every strategy $\sigma$ of Player~1, there exists an optimal counter strategy of Player~2.
\end{lemma}

\begin{proof}
For any fixed strategy $\sigma$ of Player~1, Player~2's problem is a finite discounted deterministic MDP. The conclusion therefore follows from the existence of an optimal counter strategy in \cite[Section~2]{postye2015}.
\end{proof}

We first determine the action of an optimal counter strategy at the state selected by Player~1 and then bound the cost-to-go values at each level. These properties will be used to compare the two values at each level and determine the remaining actions.

\begin{lemma}\label{lem:selected-actions}
Consider the TBDG $\mcG_h$ defined in Subsection~\ref{subsec:construction}. For any strategy $\sigma$ of Player 1 and any optimal counter strategy $\tau$ of Player 2, let $\pi=(\sigma,\tau)$. For every $1\leq i\leq h$, we have $\pi_{p_i}=(p_i,a_{i+1})$ if $\sigma_{a_i}=(a_i,p_i)$ and $\pi_{q_i}=(q_i,b_{i+1})$ otherwise.
\end{lemma}

\begin{proof}

Fix a strategy $\sigma$ of Player~1 and let $\tau$ be any optimal counter strategy of Player~2, whose existence follows from Lemma~\ref{lem:existence}. We prove the lemma by backward induction, with the terminal values given by \eqref{eq:terminal-values}. Fix $1\leq i\leq h$ and suppose for induction that the conclusion holds at every level $j=i+1,\ldots,h$.

By the induction hypothesis and the construction of $\mcG_h$, every trajectory under $\pi$ starting from $a_{i+1}$ or $b_{i+1}$ reaches $a_{h+1}$, and its only possible negative cost is $-R$, incurred at most once. Thus \eqref{eq:value} and $0<\gamma<1$ yields $v_{a_{i+1}}^\pi,v_{b_{i+1}}^\pi\geq-R$. For $i=h$, these bounds follow directly from \eqref{eq:terminal-values}.

Suppose that $\sigma_{a_i}=(a_i,p_i)$. By \eqref{eq:value}, the optimality of $\tau$ at $p_i$, and the construction of $\mcG_h$, we have
\[
v_{a_i}^\pi=c_{(a_i,p_i)}+\gamma v_{p_i}^\pi\geq c_{(a_i,p_i)}+\gamma c_{(p_i,a_{i+1})}+\gamma^2v_{a_{i+1}}^\pi=W_i+\beta v_{a_{i+1}}^\pi\geq W_i-\beta R>-R,
\]
where the first equality follows from \eqref{eq:value} and $\sigma_{a_i}=(a_i,p_i)$, the first inequality follows from the optimality of $\tau$ and $\gamma>0$, the second equality follows from $c_{(a_i,p_i)}=0$, $c_{(p_i,a_{i+1})}=W_i/\gamma$, and $\beta=\gamma^2$, the second inequality follows from $v_{a_{i+1}}^\pi\geq-R$, and the last inequality follows from \eqref{eq:separation-a}. Moreover, suppose for contradiction that $\tau_{p_i}=(p_i,a_i)$, then the selected actions form a two-action cycle between $a_i$ and $p_i$, with costs $0$ and $-K_i/\gamma$, respectively. It then follows that $v_{a_i}^\pi=\gamma v_{p_i}^\pi=-K_i/(1-\beta)<-3R$, where the equalities follow from \eqref{eq:value}, the cycle costs $0$ and $-K_i/\gamma$, and $\beta=\gamma^2$, and the inequality follows from \eqref{eq:separation-b} and \eqref{eq:separation-c}. This contradicts the lower bound $v_{a_i}^\pi>-R$ derived above. Thus we have $\tau_{p_i}=(p_i,a_{i+1})$.

Suppose that $\sigma_{a_i}=(a_i,q_i)$, the same argument applies with $q_i$ in place of $p_i$ and $b_{i+1}$ in place of $a_{i+1}$. Thus we also have $\tau_{q_i}=(q_i,b_{i+1})$, which completes the induction.
\end{proof}

\begin{lemma}\label{lem:values}
Consider the TBDG $\mcG_h$ defined in Subsection~\ref{subsec:construction}. For any strategy $\sigma$ of Player 1 and any optimal counter strategy $\tau$ of Player 2, let $\pi=(\sigma,\tau)$. Then we have $|v_{a_i}^\pi|,|v_{b_i}^\pi|\leq R$ for all $1\leq i\leq h+1$.
\end{lemma}

\begin{proof}

The terminal values are given by \eqref{eq:terminal-values}. Moreover, it follows from Lemma~\ref{lem:selected-actions} that every trajectory under $\pi$ starting from a level state finally reaches $a_{h+1}$, and its only possible negative cost is $-R$, incurred at most once. Hence \eqref{eq:value} and $0<\gamma<1$ imply $v_{a_i}^\pi,v_{b_i}^\pi\geq-R$ for every $1\leq i\leq h+1$. We next prove that
\[
v_{a_i}^\pi,v_{b_i}^\pi\leq\sum_{j=i}^hW_j,
\]
for every $1\leq i\leq h+1$ by backward induction. Fix $1\leq i\leq h$ and suppose for induction that the bounds hold at level $i+1$. 

Suppose that $\sigma_{a_i}=(a_i,p_i)$, it follows from Lemma~\ref{lem:selected-actions} that $\tau_{p_i}=(p_i,a_{i+1})$. The optimality of $\tau$ at $b_i$ together with \eqref{eq:value} and the construction of $\mcG_h$ yields
\begin{align}
&v_{a_i}^\pi=c_{(a_i,p_i)}+\gamma c_{(p_i,a_{i+1})}+\gamma^2v_{a_{i+1}}^\pi= W_i+\beta v_{a_{i+1}}^\pi,\label{eq:recursive1}\\
&v_{b_i}^\pi=\max_{y\in\{a_{i+1},b_{i+1}\}}\{c_{(b_i,y)}+\gamma v_y^\pi\}=\eta\max\{v_{a_{i+1}}^\pi,v_{b_{i+1}}^\pi\},\label{eq:recursive2}
\end{align}
where the first equality follows from \eqref{eq:value}, $\sigma_{a_i}=(a_i,p_i)$ and $\tau_{p_i}=(p_i,a_{i+1})$, the second equality follows from $c_{(a_i,p_i)}=0$, $c_{(p_i,a_{i+1})}=W_i/\gamma$ and $\beta=\gamma^2$, the third equality follows from the optimality of $\tau$ and \eqref{eq:value}, and the last equality follows from $c_{(b_i,a_{i+1})}=c_{(b_i,b_{i+1})}=0$ and $\eta=\gamma$. By the above qualities and the induction hypothesis, we have
\begin{align*}
&v_{a_i}^\pi=W_i+\beta v_{a_{i+1}}^\pi\leq W_i+\beta\sum_{j=i+1}^hW_j\leq\sum_{j=i}^hW_j,\\
&v_{b_i}^\pi=\eta\max\{v_{a_{i+1}}^\pi,v_{b_{i+1}}^\pi\}\leq\eta\sum_{j=i+1}^hW_j\leq\sum_{j=i}^hW_j,
\end{align*}
where the first inequality in each line follows from the induction hypothesis, and the last inequalities follow from $W_i>0$, $0<\beta\leq\eta<1$, and $\sum_{j=i+1}^hW_j\geq0$. The same argument applies to the case where $\sigma_{a_i}=(a_i,q_i)$, with $p_i$ and $a_{i+1}$ replaced by $q_i$ and $b_{i+1}$, respectively, in \eqref{eq:recursive1}, while \eqref{eq:recursive2} remains unchanged. Hence the induction is complete.

In addition, \eqref{eq:separation-b} and \eqref{eq:separation-c} imply
\[
\sum_{j=i}^hW_j<\frac43W_h\leq\frac R3<R\qquad \forall 1\leq i\leq h,
\]
where the first inequality follows from $W_j\leq W_h/4^{h-j}$ and the sum of a geometric series, the second inequality follows from $4W_h\leq R$, and the last inequality follows from $R>0$. This together with the established upper and lower bounds and \eqref{eq:terminal-values} implies $|v_{a_i}^\pi|,|v_{b_i}^\pi|\leq R$ for every $1\leq i\leq h+1$, which completes the proof.
\end{proof}

The next lemma determines the magnitude and sign of the value difference at each level.

\begin{lemma}\label{lem:signs}
Consider the TBDG $\mcG_h$ defined in Subsection~\ref{subsec:construction}. For any strategy $\sigma$ and any optimal counter strategy $\tau$, let $\pi=(\sigma,\tau)$. For all $1\leq i\leq h$, $|\Delta_i^\pi|>3W_i/4$ and $\operatorname{sgn}\Delta_i^\pi=\operatorname{sgn}\Delta_{i+1}^\pi$ if $\sigma_{a_i}=(a_i,p_i)$, and $\operatorname{sgn}\Delta_i^\pi=-\operatorname{sgn}\Delta_{i+1}^\pi$ otherwise.
\end{lemma}

\begin{proof}
We prove the conclusions by backward induction, using \eqref{eq:terminal-values} at the terminal level. Fix $i$ and suppose for induction that they hold for every $j=i+1,\ldots,h$.

Suppose that $\sigma_{a_i}=(a_i,p_i)$, it then follows from Lemma~\ref{lem:selected-actions} that $\tau_{p_i}=(p_i,a_{i+1})$. By the definition of $\Delta_i^\pi$, \eqref{eq:recursive1}, and \eqref{eq:recursive2}, we have
\begin{equation}\label{eq:difference}
\Delta_i^\pi=v_{a_i}^\pi-v_{b_i}^\pi=W_i+\beta v_{a_{i+1}}^\pi-\eta\max\{v_{a_{i+1}}^\pi,v_{b_{i+1}}^\pi\}=W_i+\beta[\Delta_{i+1}^\pi]_--(\eta-\beta)\max\{v_{a_{i+1}}^\pi,v_{b_{i+1}}^\pi\},
\end{equation}
where the first equality is the definition of $\Delta_i^\pi$ and the second equality follows from \eqref{eq:recursive1} and \eqref{eq:recursive2}. Using Lemma~\ref{lem:values}, the induction hypothesis, \eqref{eq:separation-a}, \eqref{eq:separation-b}, and~\eqref{eq:separation-c}, one has
\[
\bigl|(\eta-\beta)\max(v_{a_{i+1}}^\pi,v_{b_{i+1}}^\pi)\bigr|\leq(\eta-\beta)R<W_1/4\leq W_i/4,\quad\beta|\Delta_{i+1}^\pi|>3\beta W_i>\frac83W_i,
\]
where the first inequality follows from $\eta\geq\beta$ and Lemma~\ref{lem:values}, the second inequality follows from \eqref{eq:separation-c}, and the third inequality follows from $W_1\leq W_i$ by \eqref{eq:separation-b}. The fourth inequality follows from $|\Delta_{i+1}^\pi|>3W_{i+1}/4\geq3W_i$ by the induction hypothesis and \eqref{eq:separation-b} when $i<h$, and from $\Delta_{h+1}^\pi=R\geq4W_h>3W_h$ by \eqref{eq:terminal-values} and \eqref{eq:separation-c} when $i=h$. The last inequality follows from $\beta>8/9$ by \eqref{eq:separation-a}. By the preceding bounds and \eqref{eq:difference}, we have
\[
\left\{\begin{aligned}
&\Delta_i^\pi>\frac34W_i,\quad\mbox{if}\quad \Delta_{i+1}^\pi>0,\\
&\Delta_i^\pi<\frac54W_i-\frac83W_i=-\frac{17}{12}W_i<-\frac34W_i,\quad\mbox{if}\quad\Delta_{i+1}^\pi<0.
\end{aligned}\right.
\]
Indeed, $[\Delta_{i+1}^\pi]_-=0$ in the first case and $\beta[\Delta_{i+1}^\pi]_-<-8W_i/3$ in the second case, while the absolute value of the last term in \eqref{eq:difference} is less than $W_i/4$. Thus $|\Delta_i^\pi|>3W_i/4$ and $\operatorname{sgn}\Delta_i^\pi=\operatorname{sgn}\Delta_{i+1}^\pi$.

Now suppose that $\sigma_{a_i}=(a_i,q_i)$. Lemma~\ref{lem:selected-actions} yields $\tau_{q_i}=(q_i,b_{i+1})$, and \eqref{eq:value} yields $v_{a_i}^\pi=W_i+\beta v_{b_{i+1}}^\pi$. Hence $[\Delta_{i+1}^\pi]_-$ in \eqref{eq:difference} is replaced by $[-\Delta_{i+1}^\pi]_-$, and the same estimates therefore imply $|\Delta_i^\pi|>3W_i/4$ and $\operatorname{sgn}\Delta_i^\pi=-\operatorname{sgn}\Delta_{i+1}^\pi$. This completes the induction.
\end{proof}

We next use these value differences to further characterize the optimal counter strategy and establish its uniqueness.
\begin{lemma}\label{lem:counter strategy}
Consider the TBDG $\mcG_h$ defined in Subsection~\ref{subsec:construction}. For every strategy $\sigma$, the optimal counter strategy $\pi^2(\sigma)$ is unique, let $\pi=(\sigma,\pi^2(\sigma))$. For all $1\leq i\leq h$ and $s\in\{p_i,q_i\}$ with $\sigma_{a_i}\neq(a_i,s)$, we have $\pi_s=(s,a_i)$ if and only if $\Delta_i^\pi>0$.
\end{lemma}

\begin{proof}
Let $\tau$ be any optimal counter strategy against $\sigma$ and set $\pi=(\sigma,\tau)$. By Lemma~\ref{lem:signs}, \eqref{eq:separation-a}, \eqref{eq:separation-b}, \eqref{eq:separation-c}, and \eqref{eq:terminal-values}, we have
\[
\beta|\Delta_{i+1}^\pi|>3\beta W_i>\frac83W_i\geq\frac43K_i\qquad\forall1\leq i\leq h,
\]
where the first inequality follows from $|\Delta_{i+1}^\pi|>3W_{i+1}/4\geq3W_i$ by Lemma~\ref{lem:signs} and \eqref{eq:separation-b} when $i<h$, and from $\Delta_{h+1}^\pi=R\geq4W_h>3W_h$ by \eqref{eq:terminal-values} and \eqref{eq:separation-c} when $i=h$. The second and third inequalities follow from $\beta>8/9$ and $K_i\leq2W_i$ by \eqref{eq:separation-a}, respectively. Moreover, Lemma~\ref{lem:values}, \eqref{eq:separation-b}, and \eqref{eq:separation-c} imply
\[
(1-\beta)|v_{a_i}^\pi|\leq(1-\beta)R<K_h/3\leq K_i/3,
\]
where the first inequality follows from $|v_{a_i}^\pi|\leq R$, the second inequality follows from \eqref{eq:separation-c}, and the last inequality follows from \eqref{eq:separation-b}. Adding $K_i$ to $-K_i/3<(1-\beta)v_{a_i}^\pi<K_i/3$ yields
\begin{equation}\label{eq:k-margin}
\frac23K_i<K_i+(1-\beta)v_{a_i}^\pi<\frac43K_i.
\end{equation}
In particular, $0<K_i+(1-\beta)v_{a_i}^\pi<4K_i/3<\beta|\Delta_{i+1}^\pi|$ by the preceding bounds.

Fix $1\leq i\leq h$ and suppose that $\sigma_{a_i}=(a_i,p_i)$, it follows from Lemma~\ref{lem:selected-actions} that $\tau_{p_i}=(p_i,a_{i+1})$. By \eqref{eq:value} and the construction of $\mcG_h$, we have $v_{a_i}^\pi=\gamma v_{p_i}^\pi=W_i+\beta v_{a_{i+1}}^\pi$. This together with \eqref{eq:reduced} yields
\begin{align*}
&\gamma(r_{(q_i,b_{i+1})}^\pi-r_{(q_i,a_i)}^\pi)=\gamma c_{(q_i,b_{i+1})}-\gamma c_{(q_i,a_i)}+\gamma^2 v_{b_{i+1}}^\pi-\gamma^2 v_{a_i}^\pi\\
&=W_i+K_i+\beta v_{b_{i+1}}^\pi-\beta v_{a_i}^\pi=-\beta\Delta_{i+1}^\pi+K_i+(1-\beta)v_{a_i}^\pi,
\end{align*}
where the first equality follows from \eqref{eq:reduced}, the second equality follows from the construction of $\mcG_h$ and $\beta=\gamma^2$, and the last equality follows from \eqref{eq:recursive1} and $\Delta_{i+1}^\pi=v_{a_{i+1}}^\pi-v_{b_{i+1}}^\pi$. The preceding bounds imply that this difference is negative when $\Delta_{i+1}^\pi>0$ and positive when $\Delta_{i+1}^\pi<0$. By the optimality of $\tau$, one has $\tau_{q_i}=(q_i,a_i)$ if and only if $\Delta_{i+1}^\pi>0$, or equivalently, $\Delta_i^\pi>0$ by Lemma~\ref{lem:signs}.

Now suppose that $\sigma_{a_i}=(a_i,q_i)$. The same calculation yields
\[
\gamma\bigl(r_{(p_i,a_{i+1})}^\pi-r_{(p_i,a_i)}^\pi\bigr)
=\beta\Delta_{i+1}^\pi+K_i+(1-\beta)v_{a_i}^\pi,
\]
where the sign of the first term changes because $v_{a_i}^\pi=W_i+\beta v_{b_{i+1}}^\pi$ by \eqref{eq:value} and $\tau_{q_i}=(q_i,b_{i+1})$ from Lemma~\ref{lem:selected-actions}. The same bounds imply that the difference is negative when $\Delta_{i+1}^\pi<0$ and positive when $\Delta_{i+1}^\pi>0$. Thus $\tau_{p_i}=(p_i,a_i)$ if and only if $\Delta_{i+1}^\pi<0$, or equivalently, $\Delta_i^\pi>0$ by Lemma~\ref{lem:signs}. This proves the asserted choices for every optimal counter strategy.

Since all optimal counter strategies have the same value vector and Lemma~\ref{lem:selected-actions} and the nonzero reduced-cost differences at $p_i$ and $q_i$ established above uniquely determine $\tau_{p_i}$ and $\tau_{q_i}$ for every $1\leq i\leq h$, the action at $b_i$ is unique because $\Delta_{i+1}^\pi\neq0$ by Lemma~\ref{lem:signs} and \eqref{eq:terminal-values}. Therefore the optimal counter strategy is unique. Hence the proof is complete.
\end{proof}

Lemmas~\ref{lem:selected-actions} and~\ref{lem:counter strategy} explain the cycle creation and cancellation in our construction. An unselected return action can become part of a cycle after Player~1 switches. The optimal counter strategy then changes that return action, so the cycle disappears. The sign relation in Lemma~\ref{lem:signs} determines when such a switch is improving and will govern the subsequent strategy sequence.

\needspace{12\baselineskip}
\begin{lemma}\label{lem:pivot}
Consider the TBDG $\mcG_h$ defined in Subsection~\ref{subsec:construction}. For any strategy $\sigma$ of Player~1, let $\pi=(\sigma,\pi^2(\sigma))$. For all $1\leq i\leq h$ and $s\in\{p_i,q_i\}$ such that $\sigma_{a_i}\neq(a_i,s)$, we have
\begin{equation}\label{eq:actual-reduced}
r_{(a_i,s)}^\pi=
\begin{cases}
-K_i-(1-\beta)v_{a_i}^\pi,&\Delta_i^\pi>0,\\
\beta|\Delta_{i+1}^\pi|,&\Delta_i^\pi<0.
\end{cases}
\end{equation}
If an improving action exists, the unique Dantzig pivot is at
\begin{equation}\label{eq:pivot-index}
i=\min\{j:1\leq j\leq h,\ \Delta_j^\pi>0\}.
\end{equation}
\end{lemma}

\begin{proof}
Fix $1\leq i\leq h$, and let $s\in\{p_i,q_i\}$ satisfy $\sigma_{a_i}\neq(a_i,s)$. By Lemma~\ref{lem:signs}, $\Delta_i^\pi\neq0$ and suppose that $\Delta_i^\pi>0$. By Lemma~\ref{lem:counter strategy}, $\pi_s=(s,a_i)$ and therefore
\[
r_{(a_i,s)}^\pi=c_{(a_i,s)}+\gamma v_s^\pi-v_{a_i}^\pi=\gamma(-K_i/\gamma+\gamma v_{a_i}^\pi)-v_{a_i}^\pi=-K_i-(1-\beta)v_{a_i}^\pi,
\]
where the first equality follows from \eqref{eq:reduced}, the second equality follows from \eqref{eq:value}, $\pi_s=(s,a_i)$, $c_{(a_i,s)}=0$, and $c_{(s,a_i)}=-K_i/\gamma$, and the last equality follows from $\beta=\gamma^2$.

Now suppose that $\Delta_i^\pi<0$ and $\sigma_{a_i}=(a_i,p_i)$. Then $s=q_i$, and Lemmas~\ref{lem:selected-actions}, \ref{lem:signs} and~\ref{lem:counter strategy} imply $\pi_{p_i}=(p_i,a_{i+1})$, $\Delta_{i+1}^\pi<0$ and $\pi_{q_i}=(q_i,b_{i+1})$, respectively. Using these, \eqref{eq:value}, \eqref{eq:reduced} and the construction of $\mcG_h$, one has
\[
r_{(a_i,q_i)}^\pi=c_{(a_i,q_i)}+\gamma v_{q_i}^\pi-v_{a_i}^\pi=W_i+\beta v_{b_{i+1}}^\pi-(W_i+\beta v_{a_{i+1}}^\pi)=-\beta\Delta_{i+1}^\pi=\beta|\Delta_{i+1}^\pi|>0,
\]
where the first equality follows from \eqref{eq:reduced}, the second equality follows from \eqref{eq:value}, $\pi_{q_i}=(q_i,b_{i+1})$ and the construction of $\mcG_h$, and the remaining equalities and the inequality follow from the definition of $\Delta_{i+1}^\pi$, $\Delta_{i+1}^\pi<0$, and $\beta>0$.

Similarly, suppose that $\Delta_i^\pi<0$ and $\sigma_{a_i}=(a_i,q_i)$. Then $s=p_i$, and the same argument applies with $p_i$ and $q_i$ interchanged and $a_{i+1}$ and $b_{i+1}$ interchanged. In this case, Lemma~\ref{lem:signs} yields $\Delta_{i+1}^\pi>0$, and hence $r_{(a_i,p_i)}^\pi=\beta\Delta_{i+1}^\pi=\beta|\Delta_{i+1}^\pi|>0$. Thus \eqref{eq:actual-reduced} holds.

By \eqref{eq:k-margin} and \eqref{eq:actual-reduced}, $r_{(a_i,s)}^\pi<0$ if and only if $\Delta_i^\pi>0$. For any $1\leq i<j\leq h$ and $s\in\{p_i,q_i\}$, $t\in\{p_j,q_j\}$ satisfying $r_{(a_i,s)}^\pi<0$ and $r_{(a_j,t)}^\pi<0$, we have
\[
r_{(a_i,s)}^\pi< -\frac23K_i\leq-\frac43K_j<r_{(a_j,t)}^\pi<0,
\]
where the first and third inequalities follow from \eqref{eq:k-margin} and \eqref{eq:actual-reduced}, the second inequality follows from $K_i\geq2^{j-i}K_j\geq2K_j$ by \eqref{eq:separation-b}, and the last inequality holds by assumption. Each $a_i$ with $1\leq i\leq h$ has exactly one unselected action, and every selected action has reduced cost zero by \eqref{eq:value} and \eqref{eq:reduced}. The only action at $a_{h+1}$ also has reduced cost zero by \eqref{eq:reduced} and \eqref{eq:terminal-values}. Thus, if an improving action exists, the unique minimizer of $r_a^\pi$ over $a\in\mcA^1$ has source $a_i$ with $i$ given by \eqref{eq:pivot-index}.
\end{proof}

\subsection{The exponential lower bound on TBDG with unrestricted costs}\label{subsec:counter}

We next use Lemmas~\ref{lem:signs} and~\ref{lem:pivot} to analyze the strategy sequence. The following lemma establishes an identity for the sum of $2^{j-1}$ over the levels with $\Delta_j^{\pi_k}>0$. Theorem~\ref{thm:real} derives the iteration count from this identity. The pivot sequence is described in the discussion following the theorem.

\begin{lemma}\label{lem:counter}
Consider the TBDG $\mcG_h$ defined in Subsection~\ref{subsec:construction}, and initialize Algorithm~\ref{alg:ssi} by \eqref{eq:initial}. Then we have
\begin{equation}\label{eq:counter}
C_k:=\sum_{1\leq j\leq h,\ \Delta_j^{\pi_k}>0}2^{j-1}
=2^h-1-k\qquad \forall0\leq k\leq2^h-1.
\end{equation}
\end{lemma}

\begin{proof}
We prove \eqref{eq:counter} by induction on $k$. It follows from Lemma~\ref{lem:signs} that for every $1\leq j\leq h$, $\operatorname{sgn}\Delta_j^{\pi_0}=\operatorname{sgn}\Delta_{j+1}^{\pi_0}$ if $(\sigma_0)_{a_j}=(a_j,p_j)$, and $\operatorname{sgn}\Delta_j^{\pi_0}=-\operatorname{sgn}\Delta_{j+1}^{\pi_0}$ otherwise. By \eqref{eq:initial} and \eqref{eq:terminal-values}, $(\sigma_0)_{a_j}=(a_j,p_j)$ for every $1\leq j\leq h$ and $\Delta_{h+1}^{\pi_0}=R>0$, and therefore $\Delta_j^{\pi_0}>0$ for every $1\leq j\leq h$. Using this and the definition of $C_0$, we have $C_0=2^h-1$, which proves the result for $k=0$.

Suppose for induction that \eqref{eq:counter} holds for some integer $k$ with $0\leq k<2^h-1$. Then $C_k>0$, so \eqref{eq:k-margin} and \eqref{eq:actual-reduced} imply that an improving action exists. Let $i=\min\{j:1\leq j\leq h,\ \Delta_j^{\pi_k}>0\}$. By Lemma~\ref{lem:pivot} and Algorithm~\ref{alg:ssi}, the next pivot changes only $(\sigma_k)_{a_i}$. Since all the levels after level $i$ is not changed, we have
\[
v_{a_j}^{\pi_{k+1}}=v_{a_j}^{\pi_k},\qquad
v_{b_j}^{\pi_{k+1}}=v_{b_j}^{\pi_k}\qquad \forall j=i+1,\ldots,h+1.
\]
This together with Lemma~\ref{lem:signs} yields $\operatorname{sgn}\Delta_i^{\pi_{k+1}}=-\operatorname{sgn}\Delta_i^{\pi_k}=-1$. Using these, $\Delta_j^{\pi_k}<0$ for every $1\leq j<i$ and the pivot changes only $(\sigma_k)_{a_i}$, one has
\[
\operatorname{sgn}\Delta_j^{\pi_{k+1}}=
\begin{cases}
1,&1\leq j<i,\\
-1,&j=i,\\
\operatorname{sgn}\Delta_j^{\pi_k},&i<j\leq h.
\end{cases}
\]
Consequently,
\[
C_{k+1}=C_k-2^{i-1}+\sum_{j=1}^{i-1}2^{j-1}=C_k-1,
\]
which proves \eqref{eq:counter} for $k+1$ and completes the induction.
\end{proof}

We now use \eqref{eq:counter} and Proposition~\ref{prop:reduced-optimality} to establish the iteration count and the optimality of the terminal profile.

\begin{theorem}\label{thm:real}
For every integer $h\geq2$, suppose we run Algorithm \ref{alg:ssi} initialized by \eqref{eq:initial} on the game $\mcG_h$ defined in Section~\ref{subsec:construction}. Then the following statements hold.
\begin{enumerate}[label=(\roman*)]
\item The game $\mcG_h$ has $n=4h+2$ states and $m=8h+2$ actions. All costs are rational and finite. Each cost and the discount factor have $O(h)$ binary encoding length.
\item Player~2 has a unique optimal counter strategy for every strategy of Player~1. Algorithm~\ref{alg:ssi}, initialized by \eqref{eq:initial}, has a unique Dantzig pivot at every nonterminal iterate and performs exactly $2^h-1$ iterations.
\end{enumerate}
\end{theorem}

\begin{proof}
The first statement follows directly from the construction of $\mcG_h$; see Table~\ref{tab:edges} and~\eqref{eq:real-parameters}.

The uniqueness of the optimal counter strategy and of every improving Dantzig pivot follows from Lemmas~\ref{lem:counter strategy} and~\ref{lem:pivot}, respectively. By Lemma~\ref{lem:counter}, $C_k=2^h-1-k$ for every $0\leq k\leq2^h-1$. Moreover, \eqref{eq:k-margin} and \eqref{eq:actual-reduced} imply that $\min_{a\in\mcA^1}r_a^{\pi_k}<0$ if and only if $C_k>0$. Hence an improving pivot occurs for every $0\leq k<2^h-1$. At $k=2^h-1$, we have $C_k=0$, and thus $r_a^{\pi_k}\geq0$ for every $a\in\mcA^1$. Since $\tau_k$ is an optimal counter strategy against $\sigma_k$, Proposition~\ref{prop:reduced-optimality} implies $r_a^{\pi_k}\leq0$ for every $a\in\mcA^2$. The same proposition therefore implies that $\pi_k$ is optimal. In addition, $\min_{a\in\mcA^1}r_a^{\pi_k}=0$ when $k=2^h-1$, so Algorithm~\ref{alg:ssi} terminates. It therefore performs exactly $2^h-1$  iterations, which proves statement~(ii).
\end{proof}

Theorem~\ref{thm:real} immediately implies the following iteration lower bound for general TBDGs.

\begin{corollary}\label{cor:real-lower}
Algorithm~\ref{alg:ssi} has a worst-case iteration lower bound of $2^{\Omega(n)}$ on general TBDGs. In particular, it is not strongly polynomial on this class.
\end{corollary}

The analyses in~\cite[Lemma~3.5]{postye2015} and~\cite[Lemma~6]{meisunye2026} show that creating a new cycle in a deterministic MDP, or a new Player~1 cycle in a TBDFG, respectively, reduces the optimality gap by a factor of at most $1-1/n$. For deterministic MDPs, the simplex method creates a new cycle within $O(n^2m\log n)$ iterations unless it terminates (see~\cite[Lemma~3.4]{postye2015}). In a TBDFG, Player~2 cannot remove a Player~1 cycle. Our construction instead creates mixed-player cycles that are canceled by Player~2's optimal counter strategy, so no new cycle remains after a complete iteration.

Specifically, the costs $W_i$ determine the successor-value ordering, while the decreasing $K_i$ make the smallest level with $\Delta_i^{\pi_k}>0$ the unique Dantzig pivot. Before an improving pivot $(a_i,s)$, where $s\in\{p_i,q_i\}$, the unique optimal counter strategy selects $(s,a_i)$. The pivot therefore creates a two-action cycle. Its value $-K_i/(1-\beta)<-3R$ forces the next optimal counter strategy to select the action from $s$ to the next level, as established in Lemmas~\ref{lem:selected-actions} and~\ref{lem:counter strategy}. Thus every completed profile has only the terminal self-loop as a cycle. By \eqref{eq:pivot-index} and \eqref{eq:counter}, the source index of pivot $k$ is $1+\nu_2(C_{k-1})=1+\nu_2(2^h-k)=1+\nu_2(k)$, where $1\leq k\leq2^h-1$ and $\nu_2(k)=\max\{r\in\mathbb Z_{\geq0}:2^r\mid k\}$. Hence level~$i$ pivots $2^{h-i}$ times, with its two actions alternating. Exactly $2h-1$ distinct return cycles are formed, and the cycle through $(a_1,p_1)$ is formed $2^{h-2}$ times. The separated cost scales enforce the counter identity \eqref{eq:counter}, so these creations and cancellations continue for $2^h-1=2^{\Omega(n)}$ iterations even though no new cycle survives an optimal counter strategy computation.

\subsection{A two-level example of the counterexample}\label{subsec:example}

In this subsection we illustrate the construction with $h=2$, which has ten states. Choose
\begin{equation}\label{eq:example-parameters}
\gamma=\frac{99}{100},\qquad
\frac{W_1}{\gamma}=1,\quad \frac{W_2}{\gamma}=4,\quad
\frac{K_1}{\gamma}=2,\quad \frac{K_2}{\gamma}=1,\quad R=16.
\end{equation}
The parameters in \eqref{eq:example-parameters} are obtained from \eqref{eq:real-parameters} with $h=2$ by multiplying all costs by $16$ and keeping the same discount factor. Figure~\ref{fig:two-level} shows the states, actions, and costs of the resulting game, where  Every arrow is a single action and its label is the immediate cost.
\begin{figure}[H]
\centering
\begin{tikzpicture}[>=Stealth,line width=.65pt,
every node/.style={font=\small},
state/.style={draw,minimum width=9mm,minimum height=9mm,inner sep=0pt},
p1/.style={state,rectangle},
p2/.style={state,circle},
aux/.style={p2},
pathL/.style={->},
pathB/.style={->},
cost/.style={fill=white,inner sep=1.5pt}]
\node[p1] (a1) at (0,0) {$a_1$};
\node[p2] (b1) at (6,0) {$b_1$};
\node[aux] (auxp1) at (-3,0) {$p_1$};
\node[aux] (auxq1) at (3,0) {$q_1$};
\node[p1] (a2) at (0,-1.8) {$a_2$};
\node[p2] (b2) at (6,-1.8) {$b_2$};
\node[aux] (auxp2) at (-3,-1.8) {$p_2$};
\node[aux] (auxq2) at (3,-1.8) {$q_2$};
\node[p1] (a3) at (0,-3.6) {$a_3$};
\node[p2] (b3) at (6,-3.6) {$b_3$};
\draw[pathB] (b1) -- node[cost,pos=.22,above] {$0$} (a2);
\draw[pathB] (b1) -- node[cost,right] {$0$} (b2);
\draw[pathB] (b2) -- node[cost,pos=.22,above] {$0$} (a3);
\draw[pathB] (b2) -- node[cost,right] {$0$} (b3);
\draw[->] (a1) -- node[cost,below] {$0$} (auxp1);
\draw[->] (a1) -- node[cost,below] {$0$} (auxq1);
\draw[pathL] (auxp1) to[bend left=20]
node[cost,pos=.3,above] {$-2$} (a1);
\draw[pathL] (auxq1) to[bend right=20]
node[cost,pos=.3,above] {$-2$} (a1);
\draw[pathL] (auxp1) -- node[cost,left] {$1$} (a2);
\draw[pathL] (auxq1) -- node[cost,pos=.68,below] {$1$} (b2);
\draw[->] (a2) -- node[cost,below] {$0$} (auxp2);
\draw[->] (a2) -- node[cost,below] {$0$} (auxq2);
\draw[pathL] (auxp2) to[bend left=20]
node[cost,pos=.3,above] {$-1$} (a2);
\draw[pathL] (auxq2) to[bend right=20]
node[cost,pos=.3,above] {$-1$} (a2);
\draw[pathL] (auxp2) -- node[cost,left] {$4$} (a3);
\draw[pathL] (auxq2) -- node[cost,pos=.68,below] {$4$} (b3);
\draw[pathL] (b3) -- node[cost,above] {$-16$} (a3);
\draw[->] (a3) to[loop left,looseness=6,min distance=11mm]
node[cost,left] {$0$} (a3);
\end{tikzpicture}
\caption{The ten-state example with $\gamma=0.99$.}
\label{fig:two-level}
\end{figure}
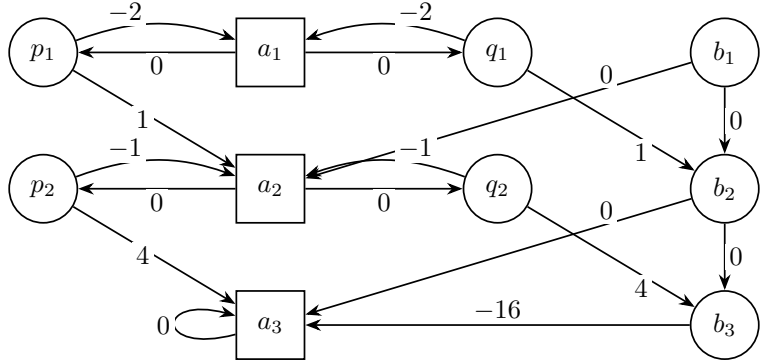

Starting from $\sigma_{a_1}=(a_1,p_1)$ and $\sigma_{a_2}=(a_2,p_2)$, Algorithm~\ref{alg:ssi} proceeds as shown in Table~\ref{tab:example-iterations}. Player~2 uses an optimal counter strategy at initialization and after each pivot. The details are shown in the following table.
\begin{table}[H]
\centering\small
\begin{tabular}{cccc}
\toprule
$k$ & $(\sigma_k)_{a_1}$ & $(\sigma_k)_{a_2}$ & Next entering action\\
\midrule
0 & $(a_1,p_1)$ & $(a_2,p_2)$ & $(a_1,q_1)$\\
1 & $(a_1,q_1)$ & $(a_2,p_2)$ & $(a_2,q_2)$\\
2 & $(a_1,q_1)$ & $(a_2,q_2)$ & $(a_1,p_1)$\\
3 & $(a_1,p_1)$ & $(a_2,q_2)$ & Terminate\\
\bottomrule
\end{tabular}
\caption{The strategy sequence and Dantzig pivots in the ten-state example.}
\label{tab:example-iterations}
\end{table}

The three pivots create two-action cycles through $a_1,q_1$, through $a_2,q_2$, and through $a_1,p_1$, respectively. Player~2's subsequent optimal counter strategies select $(q_1,b_2)$, $(q_2,b_3)$, and $(p_1,a_2)$, respectively, and thus cancel these cycles. Hence the initially selected action $(a_1,p_1)$ reenters at the third pivot, after which the algorithm terminates. Every completed profile has only the terminal self-loop at $a_3$ as a cycle.

\section{A superpolynomial lower bound with unit costs}\label{sec:ternary}

We now consider TBDGs with unit weights, namely, action costs in $\{-1,0,1\}$. These costs are natural when the objective counts events. For example, in quantitative controller synthesis against an adversarial environment \cite{bansaletal2021}, assigning cost $1$ to each undesirable transition and cost $0$ to every other transition measures the discounted number of undesirable events. The restriction also permits improved algorithmic guarantees. A strategy-iteration method solves this class in deterministic $n^{O(\sqrt n)}$ time, even when $\gamma$ is arbitrarily close to one (Asadi et al.~\cite{asadietal2024}). It is therefore natural to ask whether unit weights ensure a strongly polynomial bound for Algorithm~\ref{alg:ssi}. In this section, we show that the cycle creation and cancellation mechanism in Section~\ref{sec:real} still yields a superpolynomial iteration lower bound.

We extend the construction in Section~\ref{sec:real} by replacing each action with unrestricted cost by a path whose internal states have one available action. We first construct a polynomial that determines the path costs, then verify the conditions needed to apply the counter strategy and reduced-cost analysis from Subsection~\ref{subsec:comparisons}.

\subsection{The construction of the $h$-level counterexample with unit weights}\label{subsec:paths}

We first construct a polynomial with coefficients in $\{-1,0,1\}$ and a high-order zero at one. Evaluating the polynomial at suitable powers of $\gamma$ then produces costs $W_i$ increasing with $i$ and magnitudes $K_i$ decreasing with $i$ such that \eqref{eq:separation-a}, \eqref{eq:separation-b} and \eqref{eq:separation-c} hold. We extend the paths with zero-cost actions so that all paths from one level to the next and all return cycles have the same total discount, and hence $\beta=\eta$.

For all $h\geq2$, we set
\begin{equation}\label{eq:qd}
\theta=8h,\qquad D=\theta^3.
\end{equation}
We follow a similar polynomial construction in~\cite[Theorem~2.7]{borweinetal1999} and~\cite[Lemma~4.8]{asadietal2024} to construct our desired polynomial below.
\begin{lemma}\label{lem:polynomial}
For $\theta,D$ in \eqref{eq:qd}, there exist a polynomial function
\begin{equation}\label{eq:polynomial}
F(x)=\sum_{j=0}^D f_jx^j,\qquad f_j\in\{-1,0,1\},
\end{equation}
an integer $\rho$ with $\theta\leq\rho\leq D$, and a constant $\kappa>0$ such that
\[
\left|\frac{F(1-t)}{\kappa t^\rho}-1\right|\leq2^{D+2}t,\qquad 0<t\leq1/2.
\]
\end{lemma}

\begin{proof}
For every choice of $u_0,\ldots,u_D\in\{0,1\}$ and every integer $0\leq k<\theta$, the sum $\sum_{j=0}^D u_j\binom jk$ is an integer between $0$ and $\binom{D+1}{k+1}$. Hence the number of possible combinations of values of these $\theta$ sums is at most
\begin{equation}\label{eq:pigeonhole}
\prod_{k=0}^{\theta-1}\left(1+\binom{D+1}{k+1}\right)
\leq2^\theta(D+1)^{\theta(\theta+1)/2}
\leq2^{\theta+\theta^2(\theta+1)/2}<2^{D+1},
\end{equation}
where the last two inequalities follow from $D=\theta^3$ and $\theta=8h\geq16$. Since there are $2^{D+1}$ choices of $u_0,\ldots,u_D$, the pigeonhole principle implies that two distinct choices yield the same sum for every $0\leq k<\theta$. Let $f_j$ be the difference between their $j$th entries for $0\leq j\leq D$. Then the polynomial $F$ in \eqref{eq:polynomial} is nonzero and satisfies
\[
F^{(k)}(1)=k!\sum_{j=k}^D f_j\binom jk=0\qquad\forall 0\leq k<\theta.
\]
Let $\rho$ be the order of the zero of $F$ at $1$. Then $\theta\leq\rho\leq D$. Replacing $F$ by $-F$ if necessary, we may assume that $\kappa:=(-1)^\rho F^{(\rho)}(1)/\rho!\geq1$, since $F^{(\rho)}(1)/\rho!$ is a nonzero integer. Moreover, for every integer $j$ with $0\leq j\leq D$, one has
\[
\frac{|F^{(j)}(1)|}{j!}\leq\sum_{s=j}^D\binom sj=\binom{D+1}{j+1}\leq2^{D+1}.
\]
Taylor's formula and $\kappa\geq1$ therefore imply, for every $0<t\leq1/2$,
\begin{equation}\label{eq:polynomial-error}
\left|\frac{F(1-t)}{\kappa t^\rho}-1\right|
\leq2^{D+1}\sum_{j=1}^{\infty}t^j\leq2^{D+2}t,
\end{equation}
where the last inequality follows from $t/(1-t)\leq2t$.
\end{proof}

Fix one such polynomial. For example, choose the lexicographically first collision in \eqref{eq:pigeonhole}, with binary vectors ordered lexicographically, and then fix the sign. This specifies a family, although the argument does not give a polynomial-time procedure for finding the collision. Existence of the coefficient list is sufficient for a worst-case iteration lower bound.

Fix an integer $h\geq2$ and choose $F$ as in Lemma~\ref{lem:polynomial}. We construct a game $\widetilde{\mcG}_h$ from the states and actions in Subsection~\ref{subsec:construction} by replacing the actions at $p_i,q_i,b_i$, and $b_{h+1}$ with paths. All internal states on these paths are distinct, belong to Player~2, and have one available action. Specifically, for every integer $i$ with $1\leq i\leq h$, choose the parameters as follows:
\begin{equation}\label{eq:ternary-parameters}
\begin{aligned}
&B=2h+1,\qquad L=BD+1,\\
&\delta=\theta^{-4D},\qquad\gamma=1-\delta,\qquad\beta=\eta=\gamma^{L+1},\\
&W_i=\gamma F(\gamma^{h+i}),\qquad K_i=\gamma F(\gamma^{h-i+1}),
\qquad R=F(\gamma^B).
\end{aligned}
\end{equation}
For every integer $i$ with $1\leq i\leq h$, retain the zero-cost actions $(a_i,p_i)$ and $(a_i,q_i)$. Retain also the zero-cost self-loop at $a_{h+1}$. Replace each remaining action by a directed path with the same initial and terminal states. Make the following replacements, where statements~(i)--(iii) apply to every integer $i$ with $1\leq i\leq h$.
\begin{enumerate}[label=(\roman*)]
\item Replace each of $(p_i,a_{i+1})$ and $(q_i,b_{i+1})$ by a path of length $L$. For every integer $j$ with $0\leq j\leq D$, assign cost $f_j$ to the $((h+i)j+1)$-th edge of each path.
\item Replace each of $(p_i,a_i)$ and $(q_i,a_i)$ by a path of length $L$. For every integer $j$ with $0\leq j\leq D$, assign cost $-f_j$ to the $((h-i+1)j+1)$-th edge of each path.
\item Replace each of $(b_i,a_{i+1})$ and $(b_i,b_{i+1})$ by a zero-cost path of length $L+1$.
\item Replace $(b_{h+1},a_{h+1})$ by a path of length $L$. For every integer $j$ with $0\leq j\leq D$, assign cost $-f_j$ to the $(Bj+1)$-th edge.
\end{enumerate}
Assign cost $0$ to every edge whose cost is not specified above. Every specified edge position is at most $BD+1=L$, so these assignments are well defined. By \eqref{eq:polynomial} and \eqref{eq:ternary-parameters}, the discounted costs of the paths in statements~(i)--(iv) are $W_i/\gamma$, $-K_i/\gamma$, $0$, and $-R$, respectively. All individual action costs belong to $\{-1,0,1\}$. The retained action at $a_i$, followed by a path from $p_i$ or $q_i$, consists of $L+1$ actions. The two paths from $b_i$ also have length $L+1$. Hence the paths to the next level and the return cycles have total discount $\beta=\eta=\gamma^{L+1}$, as specified in \eqref{eq:ternary-parameters}. Table~\ref{tab:unit-paths} summarizes the construction, and Figure~\ref{fig:ternary-branch} illustrates the paths at level $i$.

\begin{table}[H]
\centering\small
\setlength{\tabcolsep}{3pt}
\begin{tabular}{lcccccc}
\toprule
State & Owner & Actions & Endpoint & Path length & Internal states & Discounted cost\\
\midrule
$a_i$ & Player~1 & $2$ & $p_i$; $q_i$ & $1$; $1$ & $0$; $0$ & $0$; $0$\\
$p_i$ & Player~2 & $2$ & $a_{i+1}$; $a_i$ & $L$; $L$ & $L-1$; $L-1$ & $W_i/\gamma$; $-K_i/\gamma$\\
$q_i$ & Player~2 & $2$ & $b_{i+1}$; $a_i$ & $L$; $L$ & $L-1$; $L-1$ & $W_i/\gamma$; $-K_i/\gamma$\\
$b_i$ & Player~2 & $2$ & $a_{i+1}$; $b_{i+1}$ & $L+1$; $L+1$ & $L$; $L$ & $0$; $0$\\
$a_{h+1}$ & Player~1 & $1$ & $a_{h+1}$ & $1$ & $0$ & $0$\\
$b_{h+1}$ & Player~2 & $1$ & $a_{h+1}$ & $L$ & $L-1$ & $-R$\\
Internal state & Player~2 & $1$ & Next state & $1$ & $0$ & $-1$, $0$, or $1$\\
\bottomrule
\end{tabular}
\caption{The $h$-level counterexample with unit weights.}
\label{tab:unit-paths}
\end{table}

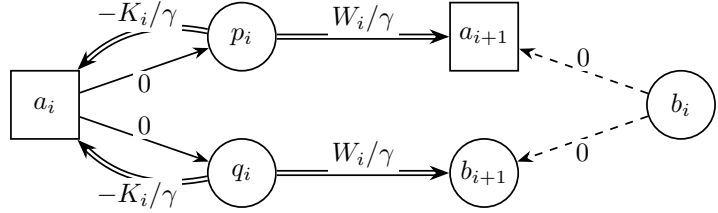
\begin{figure}[H]
\centering
\begin{tikzpicture}[>=Stealth,line width=.65pt,
every node/.style={font=\small},
state/.style={draw,minimum width=9mm,minimum height=9mm,inner sep=0pt},
p1/.style={state,rectangle},
p2/.style={state,circle},
aux/.style={p2},
pathL/.style={->,double=white,double distance=1pt,
preaction={draw=white,line width=3.2pt}},
pathB/.style={->,dashed},
cost/.style={fill=white,inner sep=1.5pt}]
\node[p1] (a) at (0,-.9) {$a_i$};
\node[aux] (p) at (2.6,0) {$p_i$};
\node[aux] (q) at (2.6,-1.8) {$q_i$};
\node[p1] (anext) at (5.8,0) {$a_{i+1}$};
\node[p2] (bnext) at (5.8,-1.8) {$b_{i+1}$};
\node[p2] (b) at (8.4,-.9) {$b_i$};
\draw[->] (a) -- node[cost,below] {$0$} (p);
\draw[->] (a) -- node[cost,above] {$0$} (q);
\draw[pathL] (p) -- node[cost,above] {$W_i/\gamma$} (anext);
\draw[pathL] (q) -- node[cost,above] {$W_i/\gamma$} (bnext);
\draw[pathL] (p) to[bend right=28] node[cost,above] {$-K_i/\gamma$} (a);
\draw[pathL] (q) to[bend left=28] node[cost,below] {$-K_i/\gamma$} (a);
\draw[pathB] (b) -- node[cost,above] {$0$} (anext);
\draw[pathB] (b) -- node[cost,below] {$0$} (bnext);
\end{tikzpicture}
\caption{The two successor branches at level $i$. In the graph, single arrows denote actions, double arrows denote $L$-edge forced paths, and dashed arrows denote $(L+1)$-edge zero-cost paths. Labels denote discounted path costs.}
\label{fig:ternary-branch}
\end{figure}
To summarize, each of the $4h$ states $a_i,b_i,p_i,q_i$, for $1\leq i\leq h$, has two actions, while both terminal states and all internal states have one. The paths from $p_i$ and $q_i$, those from $b_i$, and the path from $b_{h+1}$ introduce $4h(L-1)$, $2hL$, and $L-1$ internal states, respectively. Thus the game has $n=(6h+1)L+1$ states and $m=n+4h$ actions.

Formally, $\widetilde{\mcG}_h=(\mcS,\mcA,\bfc,\bfP,\gamma)$, where $\mcS$ consists of the original $4h+2$ states and all internal path states. Player~1 controls $\mcS^1=\{a_i:1\leq i\leq h+1\}$, and Player~2 controls $\mcS^2=\mcS\setminus\mcS^1$. The action set $\mcA$ consists of the path edges and the retained actions. Each action $a=(s,t)$ has the cost assigned above and transition column $\bfp_a=\bfe_t$ in $\bfP$. The discount factor is specified in \eqref{eq:ternary-parameters}. No other states or actions are present. We initialize Player~1's strategy by \eqref{eq:initial}.

\subsection{The superpolynomial lower bound on TBDG with unit costs}\label{subsec:ternary-estimates}

We first verify \eqref{eq:separation-a}, \eqref{eq:separation-b}, and~\eqref{eq:separation-c} for $\widetilde{\mcG}_h$, then compare the value and reduced-cost formulas for the two constructions to apply the analysis in Subsection~\ref{subsec:comparisons}.

\begin{lemma}\label{lem:ternary-separation}
Consider the TBDG $\widetilde{\mcG}_h$ defined in Subsection~\ref{subsec:paths}. Its parameters in \eqref{eq:ternary-parameters} satisfy \eqref{eq:separation-a}, \eqref{eq:separation-b}, and~\eqref{eq:separation-c}.
\end{lemma}

\begin{proof}
By \eqref{eq:qd} and \eqref{eq:ternary-parameters}, $\theta\geq16$, $B<\theta$, and $L<\theta^4$. For every integer $d$ with $1\leq d\leq B$, let $t_d=1-\gamma^d$. Observe that
\begin{equation}\label{eq:td}
1-\frac{(d-1)\delta}{2}
\leq\frac{t_d}{d\delta}
=\frac1d\sum_{j=0}^{d-1}(1-\delta)^j\leq1.
\end{equation}
Since $\delta=\theta^{-4D}$ and $D\geq5$, we also have
\begin{equation}\label{eq:small-powers}
\delta LB^D<\theta^{4-3D}\leq\theta^{-2D}.
\end{equation}
For every such $d$, $0<t_d\leq B\delta<1/2$, and \eqref{eq:td} implies $0\leq1-(t_d/(d\delta))^\rho\leq\rho(d-1)\delta/2\leq DB\delta/2$. Combining this bound with \eqref{eq:polynomial-error} yields
\begin{equation}\label{eq:uniform-approx}
\left|\frac{F(\gamma^d)}{\kappa(d\delta)^\rho}-1\right|
\leq\left(2^{D+2}+\frac D2\right)B\delta
\leq2^{D+4}B\delta<\frac1{16},
\end{equation}
where the last inequality follows from $2^{D+4}B\delta<2^{5-3D}<1/16$. Thus $F(\gamma^d)>0$ for every integer $d$ with $1\leq d\leq B$. For every integer $d$ with $1\leq d<B$, \eqref{eq:uniform-approx} implies that
\begin{equation}\label{eq:adjacent-ratio}
\frac{F(\gamma^{d+1})}{F(\gamma^d)}
>\frac{15}{17}\left(1+\frac1d\right)^\rho
\geq\frac{15}{17}\left(1+\frac\rho d\right)
\geq\frac{75}{17}>4,
\end{equation}
where we used $d\leq2h$ and $\rho\geq8h$. By \eqref{eq:ternary-parameters} and \eqref{eq:adjacent-ratio}, $W_{i+1}>4W_i$ and $K_i>4K_{i+1}$ for every integer $i$ with $1\leq i<h$. Also, $W_1>4K_1$ and $R>4W_h$. These inequalities establish \eqref{eq:separation-b} and imply $0<K_i\leq K_1<W_1/4\leq W_i/4$ for every integer $i$ with $1\leq i\leq h$.

By \eqref{eq:small-powers}, one has
\[
1-\beta\leq(L+1)\delta\leq2L\delta
\leq2\delta LB^D<2\theta^{-2D}<\frac19.
\]
Hence $8/9<\beta=\eta<1$, which, together with the positivity of $W_i,K_i,R$ and the preceding bounds on $K_i$, proves \eqref{eq:separation-a}. Since $\gamma>15/16$, \eqref{eq:ternary-parameters} and \eqref{eq:uniform-approx} yield
\begin{equation}\label{eq:r-over-k}
\frac{R}{K_h}<\frac{17}{15\gamma}B^\rho<2B^D,
\end{equation}
where the last inequality uses $\rho\leq D$. Combining \eqref{eq:small-powers} and \eqref{eq:r-over-k} with $1-\beta\leq2L\delta$ yields
\[
\frac{(1-\beta)R}{K_h}<4\delta LB^D<4\theta^{-2D}<\frac13.
\]
Together with $R>4W_h$ and $4(\eta-\beta)R=0<W_1$, this proves \eqref{eq:separation-c}. Hence the conclusion follows.
\end{proof}

Every internal state has one available action, so each path acts in the value equations as a single edge with its discounted cost and total discount. Thus the value and reduced-cost formulas have the same form as in Subsection~\ref{subsec:comparisons}, with $\beta=\eta=\gamma^{L+1}$. By Lemma~\ref{lem:ternary-separation}, Lemmas~\ref{lem:selected-actions}--\ref{lem:counter} therefore apply to $\widetilde{\mcG}_h$, with each replaced edge interpreted as its corresponding path. We omit the details.

We can now establish the iteration lower bound for Algorithm~\ref{alg:ssi} on $\widetilde{\mcG}_h$.

\begin{theorem}\label{thm:ternary}
For every integer $h\geq2$, suppose we run Algorithm~\ref{alg:ssi} initialized by \eqref{eq:initial} on the game $\widetilde{\mcG}_h$ defined in Subsection~\ref{subsec:paths}. Then the following statements hold.
\begin{enumerate}[label=(\roman*)]
\item The game $\widetilde{\mcG}_h$ has $n=(6h+1)((2h+1)(8h)^3+1)+1=\Theta(h^5)$ states and $m=n+4h$ actions. All action costs belong to $\{-1,0,1\}$. The discount factor is rational and has $O(h^3\log h)$ binary encoding length.
\item Player~2 has a unique optimal counter strategy for every strategy of Player~1. Algorithm~\ref{alg:ssi}, initialized by \eqref{eq:initial}, has a unique Dantzig pivot at every nonterminal iterate and performs exactly $2^h-1$ iterations.
\end{enumerate}
\end{theorem}

\begin{proof}
The first statement follows directly from Table~\ref{tab:unit-paths}, \eqref{eq:qd}, \eqref{eq:ternary-parameters}, and the state and action counts in Subsection~\ref{subsec:paths}, since the denominator $\theta^{4D}$ of $\gamma$ has $O(D\log \theta)=O(h^3\log h)$ bits.

By the analysis in Subsection~\ref{subsec:ternary-estimates}, Lemmas~\ref{lem:counter strategy}, \ref{lem:pivot}, and~\ref{lem:counter} apply to this game. Statement~(ii) therefore follows by the same argument as in the proof of Theorem~\ref{thm:real}(ii).
\end{proof}

Since $n=\Theta(h^5)$, we can express the iteration lower bound in Theorem~\ref{thm:ternary} in terms of the number of states as follows.

\begin{corollary}\label{cor:unit-lower}
Algorithm~\ref{alg:ssi} has a worst-case iteration lower bound of $2^{\Omega(n^{1/5})}$ on TBDGs with unit weights. In particular, it is not strongly polynomial on this class.
\end{corollary}

\subsection{A two-level example with unit weights}\label{subsec:ternary-example}

In this subsection we illustrate the construction with $h=2$. The parameters in \eqref{eq:qd} and \eqref{eq:ternary-parameters}, and the state and action counts in Subsection~\ref{subsec:paths} become
\begin{equation}\label{eq:ternary-example-parameters}
\begin{aligned}
&\theta=16,\qquad D=4096,\qquad B=5,\qquad L=20481,\\
&\gamma=1-16^{-16384},\qquad \beta=\eta=\gamma^{20482},\\
&n=13L+1=266\,254,\qquad m=n+8=266\,262.
\end{aligned}
\end{equation}
Choose $F$ as in Lemma~\ref{lem:polynomial}. The discounted path costs are determined by
\begin{equation}\label{eq:ternary-example-costs}
\begin{aligned}
&W_1=\gamma F(\gamma^3),\qquad W_2=\gamma F(\gamma^4),\\
&K_1=\gamma F(\gamma^2),\qquad K_2=\gamma F(\gamma),\qquad R=F(\gamma^5).
\end{aligned}
\end{equation}
Every action cost belongs to $\{-1,0,1\}$. Figure~\ref{fig:ternary-two-level} shows the resulting game in the same layout as Figure~\ref{fig:two-level}, with internal path states omitted. The labels are the discounted path costs specified by \eqref{eq:ternary-example-costs}. Double arrows denote $L$-action paths, and dashed arrows denote $(L+1)$-action zero-cost paths.
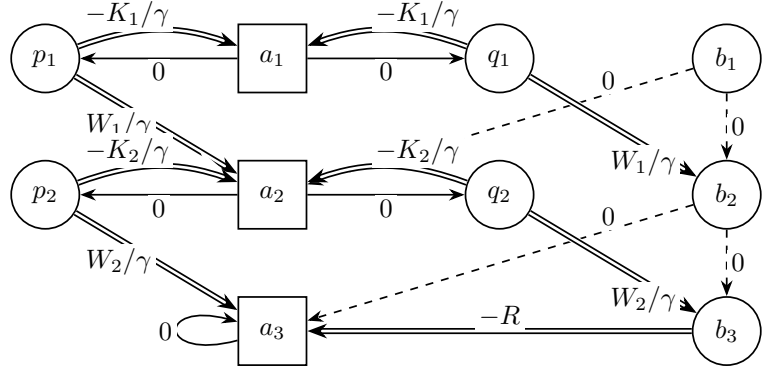
\begin{figure}[H]
\centering
\begin{tikzpicture}[>=Stealth,line width=.65pt,
every node/.style={font=\small},
state/.style={draw,minimum width=9mm,minimum height=9mm,inner sep=0pt},
p1/.style={state,rectangle},
p2/.style={state,circle},
aux/.style={p2},
pathL/.style={->,double=white,double distance=1pt,
preaction={draw=white,line width=3.2pt}},
pathB/.style={->,dashed},
cost/.style={fill=white,inner sep=1.5pt}]
\node[p1] (a1) at (0,0) {$a_1$};
\node[p2] (b1) at (6,0) {$b_1$};
\node[aux] (auxp1) at (-3,0) {$p_1$};
\node[aux] (auxq1) at (3,0) {$q_1$};
\node[p1] (a2) at (0,-1.8) {$a_2$};
\node[p2] (b2) at (6,-1.8) {$b_2$};
\node[aux] (auxp2) at (-3,-1.8) {$p_2$};
\node[aux] (auxq2) at (3,-1.8) {$q_2$};
\node[p1] (a3) at (0,-3.6) {$a_3$};
\node[p2] (b3) at (6,-3.6) {$b_3$};
\draw[pathB] (b1) -- node[cost,pos=.22,above] {$0$} (a2);
\draw[pathB] (b1) -- node[cost,right] {$0$} (b2);
\draw[pathB] (b2) -- node[cost,pos=.22,above] {$0$} (a3);
\draw[pathB] (b2) -- node[cost,right] {$0$} (b3);
\draw[->] (a1) -- node[cost,below] {$0$} (auxp1);
\draw[->] (a1) -- node[cost,below] {$0$} (auxq1);
\draw[pathL] (auxp1) to[bend left=20]
node[cost,pos=.3,above] {$-K_1/\gamma$} (a1);
\draw[pathL] (auxq1) to[bend right=20]
node[cost,pos=.3,above] {$-K_1/\gamma$} (a1);
\draw[pathL] (auxp1) -- node[cost,left] {$W_1/\gamma$} (a2);
\draw[pathL] (auxq1) -- node[cost,pos=.68,below] {$W_1/\gamma$} (b2);
\draw[->] (a2) -- node[cost,below] {$0$} (auxp2);
\draw[->] (a2) -- node[cost,below] {$0$} (auxq2);
\draw[pathL] (auxp2) to[bend left=20]
node[cost,pos=.3,above] {$-K_2/\gamma$} (a2);
\draw[pathL] (auxq2) to[bend right=20]
node[cost,pos=.3,above] {$-K_2/\gamma$} (a2);
\draw[pathL] (auxp2) -- node[cost,left] {$W_2/\gamma$} (a3);
\draw[pathL] (auxq2) -- node[cost,pos=.68,below] {$W_2/\gamma$} (b3);
\draw[pathL] (b3) -- node[cost,above] {$-R$} (a3);
\draw[->] (a3) to[loop left,looseness=6,min distance=11mm]
node[cost,left] {$0$} (a3);
\end{tikzpicture}
\caption{The construction with unit weights for $h=2$.}
\label{fig:ternary-two-level}
\end{figure}

Starting from $\sigma_{a_1}=(a_1,p_1)$ and $\sigma_{a_2}=(a_2,p_2)$, Algorithm~\ref{alg:ssi} proceeds as shown in Table~\ref{tab:unit-iterations}. Player~2 uses an optimal counter strategy at initialization and after each pivot.

\begin{table}[H]
\centering\small
\begin{tabular}{cccc}
\toprule
$k$ & $(\sigma_k)_{a_1}$ & $(\sigma_k)_{a_2}$ & Next entering action\\
\midrule
0 & $(a_1,p_1)$ & $(a_2,p_2)$ & $(a_1,q_1)$\\
1 & $(a_1,q_1)$ & $(a_2,p_2)$ & $(a_2,q_2)$\\
2 & $(a_1,q_1)$ & $(a_2,q_2)$ & $(a_1,p_1)$\\
3 & $(a_1,p_1)$ & $(a_2,q_2)$ & Terminate\\
\bottomrule
\end{tabular}
\caption{The strategy sequence and Dantzig pivots in the two-level example with unit weights.}
\label{tab:unit-iterations}
\end{table}

The three pivots create cycles through $a_1,q_1$, through $a_2,q_2$, and through $a_1,p_1$, respectively. Each cycle consists of one Player~1 action and an $L$-action return path, so its length is $L+1=20482$. Player~2's subsequent optimal counter strategies select the paths from $q_1$ to $b_2$, from $q_2$ to $b_3$, and from $p_1$ to $a_2$, respectively, and cancel these cycles. Hence $(a_1,p_1)$ reenters at the third pivot, after which the algorithm terminates. Every completed profile has only the terminal self-loop as a cycle.

\section{A Dantzig counterpart for mean-payoff TBDGs}\label{sec:mean-payoff}
We consider TBDGs with long-run average costs and use the notation of Section~\ref{sec:prelim}. For a strategy profile $\pi$, its long-run average cost from state $s$ is
\[
g_s^\pi=\lim_{N\to\infty}\frac1N\sum_{j=0}^{N-1}c_{\pi_{s_j}},\qquad\mbox{where}\quad
s_0=s,\quad \bfp_{\pi_{s_j}}=\bfe_{s_{j+1}}.
\]
Let $\mathbf g^\pi$ denote the long-run average cost vector. Player~1 minimizes the long-run average cost and Player~2 maximizes it. In this section, $\bfv^\pi$ denotes the bias vector, and $r_a^\pi$ denotes the corresponding reduced cost of action $a\in\mcA_s$. Referring the detailed analysis in \cite{zhongye2026}, we have
\begin{equation}\label{eq:mean-bias}
v_s^\pi=c_{\pi_s}-g_s^\pi+\bfp_{\pi_s}^T\bfv^\pi,\qquad r_a^\pi=c_a-g_s^\pi+\bfp_a^T\bfv^\pi-v_s^\pi.
\end{equation}
Fix an ordering of $\mcS$. For each cycle $C$ in $G^\pi$, let $s_C$ be its first state in this ordering and set $v_{s_C}^\pi=0$. These conditions uniquely determine $\bfv^\pi$ by~\eqref{eq:mean-bias}. In particular, $r_{\pi_s}^\pi=0$ for every $s\in\mcS$.

At iteration $k$ of the algorithm, let $\pi_k=(\sigma_k,\tau_k)$ denote the current strategy profile. The entering action is selected according to
\begin{equation}\label{eq:mean-rule}
(s_k,a'_k)\in\Argmin_{s\in\mcS^1,\,a\in\mcA_s}
\bigl(\bfp_a^T\mathbf g^{\pi_k}-g_s^{\pi_k},\ r_a^{\pi_k}\bigr),
\end{equation}
where the minimization is lexicographic. In particular, the algorithm is as follows.

\begin{algorithm}[H]
\caption{A Dantzig counterpart of simple strategy iteration}\label{alg:mean-ssi}
\begin{algorithmic}[1]
\REQUIRE A strategy $\sigma_0$ of Player~1.
\STATE Compute $\tau_0=\pi^2(\sigma_0)$ as in Step~4, and set $\pi_0=(\sigma_0,\tau_0)$ and $k=0$.
\STATE Choose $(s_k,a'_k)$ according to~\eqref{eq:mean-rule}. If $\bfp_{a'_k}^T\mathbf g^{\pi_k}-g_{s_k}^{\pi_k}=0$ and $r_{a'_k}^{\pi_k}=0$, terminate and return $\pi_k$.
\STATE Replace the selected action at the source of $a'_k$ by $a'_k$, obtaining $\sigma_{k+1}$.
\STATE Using Algorithm~\ref{alg:mean-counter-computation}, compute a counter strategy $\tau_{k+1}=\pi^2(\sigma_{k+1})$ that lexicographically maximizes $(g_s^{(\sigma_{k+1},\tau)},v_s^{(\sigma_{k+1},\tau)})$ over all counter strategies $\tau$, for every $s\in\mcS$. Set $\pi_{k+1}=(\sigma_{k+1},\tau_{k+1})$ and $k\leftarrow k+1$, and return to Step~2.
\end{algorithmic}
\end{algorithm}

\begin{remark}\label{rem:mean-counter}
Ties are broken arbitrarily. Algorithm~\ref{alg:mean-ssi} extends the Dantzig counterpart in Zhong and Ye~\cite[Section~5]{zhongye2026} to two-player mean-payoff TBDGs. It applies the pivoting rule in Zhong and Ye~\cite[(5.3)]{zhongye2026} at Player~1 states and computes a counter strategy in Step~4 after each switch. Algorithm~\ref{alg:mean-counter-computation} provides an explicit implementation of this counter strategy computation. Its correctness and strongly polynomial complexity are established in Lemma~\ref{lem:mean-counter-computation}.
\end{remark}

For a fixed strategy $\sigma$ of Player~1, let $\mcA_\sigma=\{\sigma_s:s\in\mcS^1\}\cup\mcA^2$ be the actions available when Player~1 uses $\sigma$. We first compute the maximum long-run average costs $\mathbf g$ and retain actions preserving these values. Subtracting $g_s$ from each action cost at $s$ makes the optimal cycles have total weight zero. Maximum-weight paths to the possible first states of these cycles determine the maximum biases, since the bias is normalized to zero at the first state of each cycle. Finally, we select actions attaining these path values and cycles whose first states have value zero, obtaining the counter strategy required in Step~4 of Algorithm~\ref{alg:mean-ssi}. The detailed procedure is presented in Algorithm~\ref{alg:mean-counter-computation}.

\begin{algorithm}[H]
\caption{Counter strategy computation for long-run average costs}\label{alg:mean-counter-computation}
\begin{algorithmic}[1]
\REQUIRE A strategy $\sigma$ of Player~1 and the fixed ordering of $\mcS$.
\STATE Compute $g_s=\max_\tau g_s^{(\sigma,\tau)}$ for every $s\in\mcS$. Let $\mcA_g=\bigcup_{s\in\mcS}\{a\in\mcA_\sigma\cap\mcA_s:\bfp_a^T\mathbf g=g_s\}$, and set $d_a=c_a-g_s$ for each such action. Determine the set $Z$ of states that can be the first state of a cycle using only actions in $\mcA_g$ with total weight zero under cost $d_a$.
\STATE For every $s\in\mcS$, compute
\[
v_s=\max\left\{\sum_{a\in P}d_a:\ P\text{ is a path from }s\text{ to }Z\text{ through }\mcA_g\right\}.
\]
Empty paths are allowed and have weight zero.
\STATE Choose $\tau$ using actions in $\mcA_g$ such that, for $\pi=(\sigma,\tau)$,
\[
v_s=d_{\pi_s}+\bfp_{\pi_s}^T\bfv\qquad\forall s\in\mcS,
\]
and $v_{s_C}=0$ for the first state $s_C$ of every cycle $C$ in $G^\pi$. The output is $\tau$.
\end{algorithmic}
\end{algorithm}

\begin{remark}\label{rem:mean-counter-computation}
\begin{enumerate}[label=(\roman*)]
\item In Step~1, $\mathbf g$ can be computed by applying Karp's maximum cycle mean algorithm~\cite{karp1978} to the strongly connected components and propagating their values backward. By the definition of $\mathbf g$, one has $g_s\geq g_t$ for every action from $s$ to $t$ in $\mcA_\sigma$, every state reaches $Z$ through $\mcA_g$, and every cycle through $\mcA_g$ has nonpositive total weight under $d_a$. To determine $Z$ in Step~1, use the Bellman-Ford algorithm~\cite{bellman1958} to compute $u_s$, the maximum total weight of a path starting at $s$ using $\mcA_g$ and weights $d_a$, with empty paths allowed. Removing cycles cannot decrease path weights, so the maximum is attained by a path of at most $n-1$ actions. By the definition of $\bfu$, one has $u_s\geq d_a+\bfp_a^T\bfu$ for every $a\in\mcA_g\cap\mcA_s$. Retain only actions $a\in\mcA_g\cap\mcA_s$ satisfying $u_s=d_a+\bfp_a^T\bfu$. Since the nonnegative quantities $u_s-d_a-\bfp_a^T\bfu$ sum to $-\sum_{a\in C}d_a$ on a cycle $C$, this cycle has total weight zero if and only if all its actions are retained. Consequently, $z\in Z$ if and only if $z$ lies on a cycle in the retained subgraph induced by the states not preceding $z$. Computing $\mathbf g$ and $\bfu$ requires $O(nm)$ operations, and identifying $Z$ requires at most $n$ graph searches of $O(n+m)$ operations each. Thus Step~1 requires $O(nm)$ operations.

\item In Step~2, $\bfv$ can also be computed using the Bellman-Ford algorithm~\cite{bellman1958} by reversing the edges associated with $\mcA_g$, negating their weights, and adding a source with zero-weight edges to every state in $Z$. By item~(i), every state reaches $Z$ and every cycle using $\mcA_g$ has nonpositive total weight. Hence $v_s$ is well-defined and equals the negative of the shortest-path distance from the added source to $s$. The computation requires $O(nm)$ operations.

\item For Step~3, consider the graph of actions $a\in\mcA_g\cap\mcA_s$ satisfying $d_a+\bfp_a^T\bfv=v_s$. In each strongly connected component with no outgoing edge to another component, choose a cycle whose first state $z$ belongs to $Z$ and satisfies $v_z=0$. Such a cycle exists because a path attaining $v_s$ from a state $s$ in the component ends at a state $z\in Z$ with $v_z=0$, and both this path and a zero-weight cycle whose first state is $z$ lie in the graph above and cannot leave the component. Select the actions on these cycles and, at every other state, the first action of a path with the fewest actions to the selected cycles. These paths are obtained by a reverse breadth-first search starting from all selected cycle states. Every state reaches a selected cycle, and each selected action outside these cycles strictly decreases the number of actions needed to reach them, so no other cycle is formed. At each Player~1 state, the only action in $\mcA_\sigma$ is $\sigma_s$, so the selected actions define a profile $\pi=(\sigma,\tau)$. Since every selected action belongs to the graph above, one has $v_s=d_{\pi_s}+\bfp_{\pi_s}^T\bfv$ for every $s\in\mcS$. Also, every cycle of $G^\pi$ is a selected cycle whose first state has value zero. Hence both requirements in Step~3 hold. The strongly connected component computation, cycle searches in disjoint components, and reverse breadth-first search require $O(n+m)$ operations in total.

\item The vectors $\mathbf g$ and $\bfv$ computed in Steps~1 and~2 are used to select $\tau$ in Step~3. By~\eqref{eq:mean-counter-candidate}, they equal $\mathbf g^{(\sigma,\tau)}$ and $\bfv^{(\sigma,\tau)}$, respectively. Algorithm~\ref{alg:mean-ssi} uses these values directly in~\eqref{eq:mean-rule}, without a separate evaluation of the returned profile.
\end{enumerate}
\end{remark}

The following lemma justifies these choices and establishes the complexity of the construction.

\begin{lemma}\label{lem:mean-counter-computation}
For every strategy $\sigma$ of Player~1, Algorithm~\ref{alg:mean-counter-computation} is well-defined and returns a counter strategy $\tau$ satisfying the lexicographic maximization requirement in Step~4 of Algorithm~\ref{alg:mean-ssi}. Moreover, it has a strongly polynomial implementation using $O(nm)$ arithmetic operations.
\end{lemma}

\begin{proof}

The well-definiteness is discussed in Remark~\ref{rem:mean-counter-computation}(i)--(iii). Let $\tau$ be its output and $\pi=(\sigma,\tau)$. Fix a state $s$, let $C$ be its eventual cycle under $\pi$, and let $z$ be the first state of $C$. Summing the equalities in Step~3 around $C$ yields $\sum_{a\in C}d_a=0$. Write $s^0=s,\ldots,s^L=z$ for its path to $z$. Then we have
\begin{align}
g_s^\pi=\frac1{|C|}\sum_{a\in C}c_a=g_s,\quad v_s^\pi=\sum_{j=0}^{L-1}(c_{\pi_{s^j}}-g_s)
=\sum_{j=0}^{L-1}(v_{s^j}-v_{s^{j+1}})
=v_s.\label{eq:mean-counter-candidate}
\end{align}
The first relation follows from the definition of long-run average cost, the fact that every selected action belongs to $\mcA_g$, and $\sum_{a\in C}d_a=0$. In the second relation, the first equality follows from~\eqref{eq:mean-bias} and $v_z^\pi=0$, the second follows from the selected-action equalities in Step~3, and the last follows by telescoping and $v_z=0$ in that step.

Fix any counter strategy $\bar\tau$ and let $\bar\pi=(\sigma,\bar\tau)$. By Step~1, one has $g_s^{\bar\pi}\leq g_s$. Suppose that equality holds at a state $s$. For every state $s'$ on its trajectory under $\bar\pi$, we have
\[
g_s\geq g_{s'}\geq g_{s'}^{\bar\pi}=g_s^{\bar\pi}=g_s,
\]
where the inequalities follow from Remark~\ref{rem:mean-counter-computation}(i) and Step~1, respectively, and the equalities follow from the definition of long-run average cost and the assumed equality at $s$. Thus every action on this trajectory belongs to $\mcA_g$, and its eventual cycle has total weight zero. Its first state $z$ therefore belongs to $Z$. Let $s^0=s,\ldots,s^L=z$ be the path under $\bar\pi$. Then
\[
v_s^{\bar\pi}=\sum_{j=0}^{L-1}(c_{\bar\pi_{s^j}}-g_s)
=\sum_{j=0}^{L-1}d_{\bar\pi_{s^j}}
\leq v_s=v_s^\pi,
\]
where the first equality follows from~\eqref{eq:mean-bias} and $v_z^{\bar\pi}=0$, the second equality follows from the definition of $d_a$ and $g_{s^j}=g_s$, the inequality follows from Step~2, and the last equality follows from~\eqref{eq:mean-counter-candidate}. This proves the lexicographic maximization requirement in Step~4 of Algorithm~\ref{alg:mean-ssi}.

By Remark~\ref{rem:mean-counter-computation}(i)--(iii), the total number of operations is $O(nm)+O(nm)+O(n+m)=O(nm)$. For rational costs, the computed values are path sums and cycle means involving polynomially many input costs, with division only by cycle lengths at most $n$. Their encoding lengths are therefore polynomial in the input length. Hence the implementation is strongly polynomial.
\end{proof}

We next define optimality for the long-run average cost and establish convergence of Algorithm~\ref{alg:mean-ssi}.

\begin{definition}[Optimality for long-run average cost]\label{def:mean-optimality}
A strategy profile $\pi^*=(\sigma^*,\tau^*)$ is optimal if
\[
g_s^{(\sigma^*,\tau)}\leq g_s^{\pi^*}\leq g_s^{(\sigma,\tau^*)}
\qquad\forall s\in\mcS,
\]
for every strategy $\sigma$ of Player~1 and every strategy $\tau$ of Player~2. Such a profile is also called an equilibrium.
\end{definition}

The next lemma establishes the long-run average cost and reduced-cost inequalities implied by the lexicographic optimality requirement in Step~4 of Algorithm~\ref{alg:mean-ssi}.

\begin{lemma}\label{lem:mean-counter}
Fix a strategy $\sigma$ of Player~1 and let $\pi=(\sigma,\pi^2(\sigma))$. For every $s\in\mcS^2$ and $a\in\mcA_s$, one has $\bfp_a^T\mathbf g^\pi\leq g_s^\pi$, and equality implies $r_a^\pi\leq0$.
\end{lemma}

\begin{proof}
Let $\bfp_a=\bfe_t$, and obtain $\pi'$ from $\pi$ by replacing only the selected action at $s$ by $a$. If $g_t^\pi>g_s^\pi$, the trajectory from $t$ under $\pi$ cannot reach $s$. Hence $g_s^{\pi'}=g_t^\pi>g_s^\pi$, contradicting optimality of $\pi^2(\sigma)$. Thus $\bfp_a^T\mathbf g^\pi=g_t^\pi\leq g_s^\pi$.

Suppose that $\bfp_a^T\mathbf g^\pi=g_t^\pi=g_s^\pi$ and $r_a^\pi>0$. If the trajectory from $t$ under $\pi$ reaches $s$, replacing $\pi_s$ by $a$ forms a cycle $C$ under $\pi'$. Summing~\eqref{eq:mean-bias} along this cycle yields $g_s^{\pi'}=g_s^\pi+r_a^\pi/|C|>g_s^\pi$. Otherwise, the eventual cycle and its first state are unchanged, so~\eqref{eq:mean-bias} yields $g_s^{\pi'}=g_s^\pi$ and $v_s^{\pi'}=c_a-g_s^\pi+v_t^\pi=v_s^\pi+r_a^\pi>v_s^\pi$. Both cases contradict the lexicographic optimality of $\pi^2(\sigma)$. Hence the conclusion follows.
\end{proof}

The next lemma establishes the progress made at each iteration.

\begin{lemma}\label{lem:mean-progress}
For every nonterminal iteration $k$ of Algorithm~\ref{alg:mean-ssi} and every $s\in\mcS$, the pair $(g_s^{\pi_{k+1}},v_s^{\pi_{k+1}})$ is lexicographically no larger than $(g_s^{\pi_k},v_s^{\pi_k})$, moreover, the inequality is strict at $s=s_k$.
\end{lemma}

\begin{proof}
Let $\pi=(\sigma,\tau)=\pi_k$ and $\pi'=(\sigma',\tau')=\pi_{k+1}$ for the simplicity of our discussion. For every selected action $\pi_s$, one has $\bfp_{\pi_s}^T\mathbf g^\pi=g_s^\pi$ and $r_{\pi_s}^\pi=0$. Since iteration $k$ is nonterminal, Step~2 and~\eqref{eq:mean-rule} imply $\bfp_{a'_k}^T\mathbf g^\pi\leq g_{s_k}^\pi$, with $r_{a'_k}^\pi<0$ whenever equality holds. Together with Lemma~\ref{lem:mean-counter} and~\eqref{eq:mean-bias}, these inequalities imply that for all $s\in\mcS$
\begin{equation}\label{eq:mean-residual}
g_{s'}^\pi\leq g_s^\pi,\qquad
r^{\pi}_{\pi'_s}=c_{\pi'_s}-g_s^\pi+v_{s'}^\pi-v_s^\pi\leq0\quad\text{if }g_{s'}^\pi=g_s^\pi,
\end{equation}
where $\bfp_{\pi'_s}=\bfe_{s'}$. Along the trajectory under $\pi'$ from $s$, the old long-run average costs are nonincreasing and are constant on its eventual cycle $C$. Then~\eqref{eq:mean-residual} yields
\begin{equation}\label{eq:mean-average-comparison}
g_s^{\pi'}=\frac1{|C|}\sum_{z\in C}c_{\pi'_z}\leq\frac1{|C|}\sum_{z\in C}g_z^\pi\leq g_s^\pi\qquad\forall s\in\mcS.
\end{equation}

Suppose that $g_s^{\pi'}=g_s^\pi$. For every state $s'$ on the trajectory under $\pi'$ from $s$, one has $g_{s'}^{\pi'}=g_s^{\pi'}$, so~\eqref{eq:mean-residual} and~\eqref{eq:mean-average-comparison} imply $g_s^\pi=g_{s'}^{\pi'}\leq g_{s'}^\pi\leq g_s^\pi$ and hence $g_{s'}^\pi=g_s^\pi$. By~\eqref{eq:mean-residual}, one has $r^{\pi}_{\pi'_z}\leq0$ for all $z\in C$. Moreover, we have
\[
\sum_{z\in C}r_{\pi'_z}^\pi=\sum_{z\in C}\bigl(c_{\pi'_z}-g_z^\pi+\bfp_{\pi'_z}^T\bfv^\pi-v_z^\pi\bigr)=\sum_{z\in C}c_{\pi'_z}-|C|g_s^\pi
=|C|\bigl(g_s^{\pi'}-g_s^\pi\bigr)=0,
\]
where the first equality follows from~\eqref{eq:mean-bias}, the second follows from $g_z^\pi=g_s^\pi$ on $C$, the third follows from $g_s^{\pi'}=|C|^{-1}\sum_{z\in C}c_{\pi'_z}$, and the last follows from $g_s^{\pi'}=g_s^\pi$. Since every summed $r_{\pi'_z}^\pi$ is nonpositive, one has $r_{\pi'_z}^\pi=0$ for every $z\in C$. Thus $(\bfp_{\pi'_z}^T\mathbf g^\pi-g_z^\pi,r_{\pi'_z}^\pi)=(0,0)$ for every $z\in C$. If $a'_k$ belonged to $C$, these equalities at $z=s_k$ would contradict Step~2 and~\eqref{eq:mean-rule}. Hence $C$ does not contain $a'_k$. Since $\sigma'$ differs from $\sigma$ only at $s_k$, one has $\pi'_{s'}=\sigma'_{s'}=\sigma_{s'}$ for every Player~1 state $s'$ on $C$. The Player~2 actions on $C$ are selected by $\tau'$, so $C$ is also a cycle under $(\sigma,\tau')$. Let $z$ be its first state in the fixed ordering. Then $g_z^{(\sigma,\tau')}=g_z^\pi$ and $v_z^{(\sigma,\tau')}=0$. The lexicographic requirement on $\tau$ in Step~4 of Algorithm~\ref{alg:mean-ssi} and equality of the long-run average costs imply $v_z^\pi\geq v_z^{(\sigma,\tau')}=0$.

Let $s^0=s,\ldots,s^L=z$ be the path under $\pi'$ to $z$. By~\eqref{eq:mean-bias} and~\eqref{eq:mean-residual},
\begin{equation}\label{eq:mean-bias-comparison}
v_s^{\pi'}
=\sum_{j=0}^{L-1}(c_{\pi'_{s^j}}-g_s^\pi)
\leq\sum_{j=0}^{L-1}(v_{s^j}^\pi-v_{s^{j+1}}^\pi)
=v_s^\pi-v_z^\pi\leq v_s^\pi,
\end{equation}
where the first equality follows from \eqref{eq:mean-bias}, $v_z^{\pi'}=0$ and $g_{s^j}^{\pi'}=g_s^{\pi'}=g_s^\pi$ along the path, the first inequality follows from the case for $g_{s^j}^\pi=g_s^\pi$ in~\eqref{eq:mean-residual}, and the last inequality uses $v_z^\pi\geq0$.

If $g_{s_k}^{\pi'}<g_{s_k}^\pi$, the pair at $s_k$ strictly decreases lexicographically.
If $g_{s_k}^{\pi'}=g_{s_k}^\pi$, then $\bfp_{a'_k}^T\mathbf g^\pi=g_{s_k}^\pi$ by~\eqref{eq:mean-residual} and~\eqref{eq:mean-average-comparison}, so~\eqref{eq:mean-rule} implies $r_{a'_k}^\pi<0$ and the first action $a'_k$ on the path from $s_k$ makes the first inequality in~\eqref{eq:mean-bias-comparison} strict, yielding $v_{s_k}^{\pi'}<v_{s_k}^\pi$. Hence we have
\begin{equation}\label{eq:mean-progress}
g_s^{\pi'}\leq g_s^\pi,\qquad
v_s^{\pi'}\leq v_s^\pi\quad\text{whenever }g_s^{\pi'}=g_s^\pi.
\end{equation}
Hence the conclusion follows.
\end{proof}

The following theorem establishes finite termination and optimality.

\begin{theorem}\label{thm:mean-convergence}
For every initial strategy $\sigma_0$, Algorithm~\ref{alg:mean-ssi} is well-defined and terminates after finitely many iterations at an optimal strategy profile in the sense of Definition~\ref{def:mean-optimality}.
\end{theorem}

\begin{proof}
Lemma~\ref{lem:mean-counter-computation} ensures that the algorithm is well-defined. By Lemma~\ref{lem:mean-progress}, every state's pair is lexicographically nonincreasing, with a strict decrease at the entering state. Therefore no profile recurs. Since the set of strategy profiles is finite, the algorithm terminates after finitely many iterations.

Let $\pi=(\sigma,\tau)$ be the returned profile. The stopping condition implies
\[
\bfp_a^T\mathbf g^\pi\geq g_x^\pi,\qquad
r_a^\pi\geq0\quad\text{whenever }\bfp_a^T\mathbf g^\pi=g_x^\pi,
\]
for every $x\in\mcS^1$ and $a\in\mcA_x$. Fix any strategy $\bar\sigma$ of Player~1 and let $\bar\pi=(\bar\sigma,\tau)$. For every $x\in\mcS^2$, one has $\bar\pi_x=\pi_x$, so $\bfp_{\bar\pi_x}^T\mathbf g^\pi=g_x^\pi$ and $r_{\bar\pi_x}^\pi=0$ by the definition of long-run average cost and~\eqref{eq:mean-bias}, respectively. Thus, along every trajectory $x_0,x_1,\ldots$ under $\bar\pi$, one has $g_{x_{t+1}}^\pi\geq g_{x_t}^\pi$ for every $t\geq0$. On the eventual cycle $C$ of the trajectory from $x$, the values $g_z^\pi$ are constant and satisfy $g_z^\pi\geq g_x^\pi$ for every $z\in C$. Consequently,
\[
g_x^{\bar\pi}=\frac1{|C|}\sum_{z\in C}c_{\bar\pi_z}\geq\frac1{|C|}\sum_{z\in C}g_z^\pi\geq g_x^\pi,
\]
where the equality follows from the definition of long-run average cost, the first inequality follows from~\eqref{eq:mean-bias} by summing $r_{\bar\pi_z}^\pi\geq0$ over $z\in C$ and canceling the bias terms around $C$, and the last inequality follows from $g_z^\pi\geq g_x^\pi$ for every $z\in C$.

Since $\tau=\pi^2(\sigma)$ also maximizes the long-run average costs, for every counter strategy $\bar\tau$ we obtain
\[
g_x^{(\sigma,\bar\tau)}\leq g_x^{(\sigma,\tau)}\leq g_x^{(\bar\sigma,\tau)}
\qquad\forall x\in\mcS.
\]
Thus $\pi$ is optimal in the sense of Definition~\ref{def:mean-optimality}.
\end{proof}

\subsection{An exponential lower bound with unrestricted costs}\label{subsec:mean-unrestricted}

We construct a counter example based on the construction of $\mcG_h$ in Subsection~\ref{subsec:construction}. In particular, replace the cost $W_i/\gamma$ on each auxiliary action to the next level by $W_i$, and the cost $-K_i/\gamma$ on each return action by $-K_i$. Retain all zero costs, the terminal cost $-R$, and the parameters
\begin{equation}\label{eq:mean-parameters}
R=1,\qquad W_i=4^{\,i-h-1},\qquad K_i=2^{\,2-i-2h}.
\end{equation}
Denote the resulting mean-payoff game by $\widehat{\mcG}_h$. The idea is unchanged: the costs $W_i$ determine the successor ordering, while the costs $-K_i$ determine the bias-pivot order. The negative return cycles are canceled by Player~2 because reaching the zero-cost sink has a larger long-run average cost.

\begin{theorem}\label{thm:mean-lower}
For every integer $h\geq2$, suppose we run Algorithm~\ref{alg:mean-ssi} initialized by \eqref{eq:initial} on the game $\widehat{\mcG}_h$ defined in Subsection~\ref{subsec:mean-unrestricted}. Then the following statements hold.
\begin{enumerate}[label=(\roman*)]
\item The game $\widehat{\mcG}_h$ has $n=4h+2$ states and $m=8h+2$ actions. All costs are rational and finite. Each cost has $O(h)$ binary encoding length.
\item Player~2 has a unique counter strategy satisfying the requirement in Step~4 of Algorithm~\ref{alg:mean-ssi} for every strategy of Player~1. Algorithm~\ref{alg:mean-ssi}, initialized by \eqref{eq:initial}, has a unique Dantzig pivot at every nonterminal iterate and performs exactly $2^h-1$ iterations.
\end{enumerate}
\end{theorem}

\begin{proof}
The first statement follows directly from the construction of $\widehat{\mcG}_h$ and~\eqref{eq:mean-parameters}.

For statement~(ii), the key difference is that every cycle other than the zero-cost self-loop at $a_{h+1}$ has negative long-run average cost. Thus, for every $\pi=(\sigma,\pi^2(\sigma))$, every trajectory reaches $a_{h+1}$, so $\mathbf g^\pi=\bfzero$ and the normalized biases are the total path costs. With these biases replacing the discounted values and the discount factors omitted, arguments similar to those in the proofs of Lemmas~\ref{lem:signs}--\ref{lem:counter} yield the same strict inequalities between reduced costs and the same pivot sequence under~\eqref{eq:mean-parameters}. Statement~(ii) therefore follows as in the proof of Theorem~\ref{thm:real}(ii).
\end{proof}

Theorem~\ref{thm:mean-lower} immediately implies the following iteration lower bound for mean-payoff TBDGs.

\begin{corollary}\label{cor:mean-lower}
Algorithm~\ref{alg:mean-ssi} has a worst-case iteration lower bound of $2^{\Omega(n)}$ on mean-payoff TBDGs with unrestricted costs. In particular, it is not strongly polynomial on this class.
\end{corollary}

For $\widehat{\mcG}_h$, the one-player method in~\cite[Theorem~7.1]{zhongye2026}, applied with $\sigma$ fixed, computes the counter strategy required in Step~4 of Algorithm~\ref{alg:mean-ssi}. Its terminal profile $\pi$ satisfies $\mathbf g^\pi=\bfzero$ and $r_a^\pi\leq0$ for every $a\in\mcA^2$. Under any counter strategy against $\sigma$, every trajectory with long-run average cost zero reaches $a_{h+1}$, where the bias is zero. Summing~\eqref{eq:mean-bias} along such a trajectory therefore shows that the returned counter strategy also maximizes the normalized bias. Thus the strongly polynomial bound in~\cite[Theorem~7.1]{zhongye2026} applies to each counter strategy computation.

\subsection{A strongly polynomial bound with unit weights}\label{subsec:mean-unit}
The restriction to unit weights provides additional structure for mean-payoff TBDGs. The long-run average cost of a trajectory depends only on its eventual cycle. With unit weights, it is a rational number in $[-1,1]$ with denominator at most $n$, so distinct long-run average costs differ by at least $1/n^2$. The normalized biases also have denominators at most $n$ and absolute values at most $2n$. These properties allow a discount with $1/(1-\gamma)$ polynomial in $n$ to preserve optimality for the long-run average cost, and bound the number of distinct pairs of long-run average costs and normalized biases. We use them to establish strongly polynomial bounds for Algorithm~\ref{alg:ssi} with discounted approximation and for Algorithm~\ref{alg:mean-ssi}, respectively.

We first consider discounted approximation. This approach uses the classical reduction in~\cite{zwickpaterson1996}, which is also discussed for strategy iteration in~\cite[Section~9]{friedmann2011}. Choose
\begin{equation}\label{eq:mean-reduction-discount}
\gamma=1-1/(8n^3).
\end{equation}

\begin{proposition}\label{prop:mean-reduction}
Consider a mean-payoff TBDG with unit weights. Then the following statements hold.
\begin{enumerate}[label=(\roman*)]
\item Every optimal strategy profile of the corresponding discounted TBDG with the discount in~\eqref{eq:mean-reduction-discount} is optimal for the mean-payoff TBDG.
\item Algorithm~\ref{alg:ssi}, applied to this discounted TBDG, returns an optimal strategy profile for the mean-payoff TBDG within $O(mn^4\log(2n))$ iterations. It has a strongly polynomial implementation, including the counter strategy computations.
\end{enumerate}
\end{proposition}

\begin{proof}
The first statement follows by applying the approximation bound in~\cite[Theorem~5.2]{zwickpaterson1996} to the game and to the one-player games obtained by fixing either player's strategy, with the discount in~\eqref{eq:mean-reduction-discount}. The second statement follows from statement~(i), the iteration bound in~\cite[Theorem~4.2]{jiawenye2020}, and the strongly polynomial counter strategy computation in~\cite{postye2015}.
\end{proof}

We next apply Algorithm~\ref{alg:mean-ssi} directly to the mean-payoff TBDG. The bounds on the long-run average costs and normalized biases limit the number of pairs that can occur at each state. Together with Lemma~\ref{lem:mean-progress}, this yields the following result.

\begin{theorem}\label{thm:mean-unit}
For a mean-payoff TBDG with unit weights, Algorithm~\ref{alg:mean-ssi} returns an optimal strategy profile within $O(n^5)$ iterations and has a strongly polynomial implementation, including the counter strategy computations.
\end{theorem}

\begin{proof}
Fix a strategy profile $\pi=(\sigma,\tau)$ and a state $s\in\mcS$. Let $C$ be the cycle eventually reached by the trajectory from $s$ under $\pi$, and let $q=|C|$ be its length. Let $s_C$ be the first state of $C$ in the fixed ordering of $\mcS$ used to normalize the bias, so $v_{s_C}^\pi=0$. Write $s_0=s,\ldots,s_L=s_C$ for the trajectory stopped at its first visit to $s_C$, where $L$ is the number of actions taken. These states are distinct, so $L\leq n-1$, and $1\leq q\leq n$. By the definition of long-run average cost and~\eqref{eq:mean-bias}, we have
\begin{equation}\label{eq:mean-cycle-bias}
g_s^\pi=\frac1q\sum_{x\in C}c_{\pi_x},\qquad
v_s^\pi=\sum_{j=0}^{L-1}(c_{\pi_{s_j}}-g_s^\pi),
\end{equation}
Define the integers $p=\sum_{x\in C}c_{\pi_x}$ and $z=q\sum_{j=0}^{L-1}c_{\pi_{s_j}}-Lp$. Then~\eqref{eq:mean-cycle-bias} implies $g_s^\pi=p/q$ and $v_s^\pi=z/q$. Since $|c_a|\leq1$ for every action $a$, we have
\[
|p|\leq q,\qquad |v_s^\pi|\leq2n,\qquad |z|\leq2nq.
\]
For each fixed $q$, there are at most $2q+1$ choices of $p$ and $4nq+1$ choices of $z$. Hence, for each state $s$, the number of possible pairs $(g_s^\pi,v_s^\pi)$ as $\pi$ varies is at most
\[
\sum_{q=1}^n(2q+1)(4nq+1)=O(n^4).
\]
For the profiles $\pi_k$ generated by Algorithm~\ref{alg:mean-ssi}, Lemma~\ref{lem:mean-progress} and~\eqref{eq:mean-progress} imply that $(g_s^{\pi_k},v_s^{\pi_k})$ never increases lexicographically at any state $s$, and strictly decreases at the state $s_k$ selected in Step~2. Thus each state can have at most $O(n^4)$ strict decreases. Since at least one state has a strict decrease in every iteration, the total number of iterations is $O(n^5)$. The returned profile is optimal by Theorem~\ref{thm:mean-convergence}.

The counter strategy computation is strongly polynomial by Lemma~\ref{lem:mean-counter-computation}. For a fixed profile, evaluating each cycle once and then propagating the values backward along the paths leading to it using~\eqref{eq:mean-cycle-bias} computes $\mathbf g^\pi$ and $\bfv^\pi$ in $O(n)$ arithmetic operations. Computing the reduced costs in~\eqref{eq:mean-bias} and selecting an action according to~\eqref{eq:mean-rule} require $O(m)$ arithmetic operations. The bounds on $p,q,z$ imply that the long-run average costs and normalized biases have $O(\log(2n))$ binary encoding length. Hence the implementation is strongly polynomial.
\end{proof}
Both approaches are therefore strongly polynomial with unit weights. Proposition~\ref{prop:mean-reduction} uses Algorithm~\ref{alg:ssi} at a fixed approximating discount, while Theorem~\ref{thm:mean-unit} bounds the iterations of Algorithm~\ref{alg:mean-ssi} directly. Strongly polynomial solvability with unit weights also follows from classical pseudopolynomial algorithms (Zwick and Paterson~\cite{zwickpaterson1996}, Brim et al.~\cite{brimetal2011}). For unit weights, the construction in Section~\ref{sec:ternary} does not preserve its strict cost separations at discount one, since $F(1)=0$.

\section{Conclusions}\label{sec:conclusion}

In this paper, we studied simple strategy iteration with Dantzig's rule and exact optimal counter strategies for discounted TBDGs. We established iteration lower bounds of $2^{\Omega(n)}$ with unrestricted costs and $2^{\Omega(n^{1/5})}$ with unit weights. The pivots and optimal counter strategies are unique in both constructions. These results show that the method is not strongly polynomial even when all action costs belong to $\{-1,0,1\}$. For TBDGs with long-run average costs, we established an exponential iteration lower bound for termination of a Dantzig counterpart with unrestricted costs. With unit weights, this counterpart admits an $O(n^5)$ iteration bound. We also showed that simple strategy iteration combined with discounted approximation has a strongly polynomial implementation for mean-payoff TBDGs with unit weights. It remains of interest to improve the lower bound for unit weights and to determine whether other pivot rules or update schemes admit strongly polynomial bounds for general discounted TBDGs.

\end{document}